%% file: main.tex
\documentclass{LMCS}
\usepackage{amsmath,amssymb}
\usepackage{hyperref,enumerate,tikz}
\usetikzlibrary{arrows}
\IfFileExists{microtype.sty}{\usepackage{microtype}}{}
\usepackage{macros-fr}
\usepackage{macros-dk}
\xdefinecolor{LemonChiffon}{rgb}{1,.98,.8}
\xdefinecolor{vert}{rgb}{0,.55,0}
\xdefinecolor{rouge}{rgb}{.9,.1,.1}
\xdefinecolor{bleu}{rgb}{0.2,0.2,0.7} 
\xdefinecolor{marron_clair}{cmyk}{0,0.46,0.50,0} 
\xdefinecolor{marron}{cmyk}{0,0.46,0.50,0.4} 
\xdefinecolor{jaune}{cmyk}{0,0,1,0} 
\xdefinecolor{violet}{cmyk}{0.47,0.91,0,0.08} 
\xdefinecolor{bleu_clair}{rgb}{0.4, 0.4, 0.9}
\xdefinecolor{bleu_vert}{cmyk}{0.62,0,0.12,0} 
\xdefinecolor{rose}{cmyk}{0,1,0,0} 

\def\letenv#1#2{\global\newenvironment{#1}{\begin{#2}}{\end{#2}}}
\letenv{theorem}{thm}
\letenv{definition}{defi}
\letenv{remark}{rem}
\letenv{lemma}{lem}
\letenv{proposition}{prop}
\letenv{corollary}{cor}
\letenv{example}{exa}

\def\doi{4 (2:9) 2008}
\lmcsheading%
{\doi}
{1--28}
{}
{}
{Sep.~20, 2007}
{Jun.~20, 2008}
{}   

\begin{document}
\title{Model Checking One-Clock Priced Timed Automata\rsuper *}

\author[P.~Bouyer]{Patricia~Bouyer\rsuper a}
\address{{\lsuper a}LSV, CNRS \& ENS de Cachan, France and 
  Oxford University Computing Laboratory, UK }
\email{bouyer@lsv.ens-cachan.fr}
\thanks{{\lsuper a}Partly supported by project DOTS (ANR-06-SETI-003), and
  by a Marie-Curie fellowship.}

\author[K.G.~Larsen]{Kim~G.~Larsen\rsuper b}
\address{{\lsuper b}Aalborg University, Denmark}
\email{kgl@cs.aau.dk}
\thanks{{\lsuper b}Partly supported by an invited professorship from ENS Cachan.}

\author[N.~Markey]{Nicolas~Markey\rsuper c}
\address{{\lsuper c}LSV, CNRS \& ENS de Cachan, France}
\email{markey@lsv.ens-cachan.fr}
\thanks{{\lsuper c}Partly supported by project DOTS (ANR-06-SETI-003)}

\keywords{priced timed automata, model-checking}
\subjclass[2000]{F.1.1,F.3.1} \titlecomment{{\lsuper *}This article is
  a long version of~\cite{BLM-fossacs07}. It has been extended with
  recent related results from~\cite{BM-formats07}.}

\input{abstract}

\maketitle

\input{introduction}

\input{definitions}

\input{modelchecking}


\input{WMTL}

\input{conclu}

\input{biblio}

\end{document}

%% file: abstract.tex
\begin{abstract}

  We consider the model of priced (a.k.a.~weighted) timed automata, an
  extension of timed automata with cost information on both locations and
  transitions, and we study various model-checking problems for that model
  based on extensions of classical temporal logics with cost constraints on
  modalities. We prove that, under the assumption that the model has only one
  clock, model-checking this class of models against the logic~\WCTL, \CTL
  with cost-constrained modalities, is \PSPACE-complete (while it has been
  shown undecidable as soon as the model has three clocks).
  We also prove that model-checking \WMTL, \LTL with cost-constrained
  modalities, is decidable only if there is a single clock in the model and a
  single stopwatch cost variable (\textit{i.e.}, whose slopes lie
  in~$\{0,1\}$).
\end{abstract}


%% file: introduction.tex
An interesting direction of real-time model-checking that has recently
received substantial attention is the extension and re-targeting of
timed automata technology towards optimal scheduling and controller
synthesis~\cite{AAM06,RLS04,BBL05}.

In particular, scheduling problems can often be reformulated in terms
of reachability questions with respect to behavioural models where
tasks and resources relevant for the scheduling problem in question
are modelled as interacting timed
automata~\cite{performance}. Although there exists a wide body of
literature and established results on (optimal) scheduling in the
fields of real-time systems and operations research, the application
of model-checking has proved to provide a novel and competitive
technology.  In~particular, model-checking has the advantage of
offering a generic approach, going well beyond most classical
scheduling solutions, which have good properties only for scenarios
satisfying specific assumptions that may or, quite often, may not
apply in actual practical circumstances. Of~course, model-checking
comes with its own restrictions and stumbling blocks, the~most
notorious being the state-space explosion. A~lot of research has thus
been devoted to ``guide'' and ``prune'' the reachability
search~\cite{guiding}.

As part of the effort on applying timed automata technology to
scheduling, the notion of priced (or weighted) timed
automata~\cite{BFH+01,ATP01}
has been promoted as a useful extension of the classical model of
timed automata allowing continuous consumption of resources
(e.g.~energy, money, pollution,~etc.) to be modelled and
analyzed. In~this way one may distinguish different feasible schedules
according to their consumption of resources
(\textit{i.e.},~accumulated cost) with obvious preference for the
\textit{optimal} schedule with minimal resource requirements.

Within the model of priced timed automata, the cost variables serve
purely as
\textit{evaluation functions} or \textit{observers}, \textit{i.e.},~the
behaviour of the underlying timed automata may in no way depend on
these cost variables. As~an important consequence of this restriction
---and in contrast to the related models of constant slope and linear
hybrid automata--- a~number of optimization problems have been shown
decidable for priced timed automata including minimum-cost
reachability~\cite{BFH+01,ATP01,BBBR07},
optimal (minimum and maximum~cost) reachability in multi-priced
settings~\cite{LR05b}
and cost-optimal infinite  schedules~\cite{BBL04,BBL05}
in terms of minimal (or~maximal) cost per time ratio in the
limit. Moreover UPPAAL Cora~\cite{cora} provides an efficient tool for
computing cost-optimal or near-optimal solutions to reachability
questions, implementing a symbolic~$A^{*}$ algorithm based on a new
data structure (so-called priced zones) allowing for efficient
symbolic state-representation with additional cost-information.

Cost-extended versions of temporal logics such as \CTL~(branching-time) and
\LTL (linear-time) appear as a natural ``generalizations'' of the above
optimization problems. Just~as \TCTL and \MTL provide extensions of
\CTL and~\LTL with time-constrained modalities, \WCTL~and \WMTL are
extensions with
\textit{cost}-constrained modalities interpreted with respect to priced timed
automata. Unfortunately, the addition of cost now turns out to come
with a price: whereas the model-checking problems for timed automata
with respect to TCTL and MTL are decidable, it has been shown
in~\cite{BBR04}
that model-checking priced timed automata with respect to \WCTL is
undecidable.  Also, in~\cite{BBR05}
it has recently been shown that the problem of determining
cost-optimal winning strategies for priced timed games is not
computable. In~\cite{BBM06}
it has been shown that these negative results hold even for priced
timed (game) automata with no more than three clocks.

Recently, the restriction of timed systems to a single clock has
raised some attention, as it leads to much nicer decidability and
complexity results.  Indeed, the emptiness problem in single-clock
timed automata becomes
\NLOGSPACE-Complete~\cite{LMS04}
instead of \PSPACE-Complete in the general framework~\cite{AD94}.
Also, the emptiness problem is decidable for single-clock alternating
timed automata and is undecidable for general alternating timed
automata~\cite{LW05,OW05,LW07,OW07}.
Even more recently, cost-optimal timed games have been proved
decidable for one-clock priced timed games~\cite{BLMR06},
and construction of almost-optimal strategies can be done.

In this paper we focus on model-checking problems for priced timed
automata with a single clock. On~the one~hand, we~show that the
model-checking problem with respect to~\WCTL is \PSPACE-Complete under
the ``single clock'' assumption. This~is rather surprising as
model-checking~\TCTL (the~only cost variable is the time elapsed)
under the same assumption is already
\PSPACE-Complete~\cite{LMS04}.
On~the other hand, we~prove that the model-checking problem with
respect to~\WMTL, the linear-time counterpart of~\WCTL, is~decidable
if we add the extra requirements that there is only one cost variable
which is stopwatch (\textit{i.e.},~with slopes in~$\{0,1\}$). We~also
prove that those two conditions are necessary to get decidability,
by~proving that any slight extension of that model leads to
undecidability.

\medskip 

The paper is organized as follows: In Section~\ref{sec:preliminaries},
we present the model of priced timed automata. Section~\ref{sec:WCTL}
is devoted to the definition of~\WCTL, and to the proof that it is
decidable when the model has only one clock. We~propose an \EXPTIME
algorithm, which we then slightly modify so that it runs
in~\PSPACE. Section~\ref{sec:WMTL} then handles the linear-time case:
we~first define~\WMTL, prove that it is decidable under the
single-clock and single-stopwatch-cost assumptions, and that it is
undecidable if we lift any of these restrictions.

%% file: definitions.tex
\section{Preliminaries}\label{sec:preliminaries}

\subsection{Priced Timed Automata}
In the sequel, $\bbbr_+$ denotes the set of nonnegative reals.
Let $\clocks$ be a set of clock variables. The set of clock constraints (or
guards) over~$\clocks$ is defined by the grammar ``$g ::= x \sim c \mid g
\wedge g$'' where $x \in \clocks$, $c \in \bbbn$ and $\sim\, \in
\{<,\leq,=,\geq,>\}$. The set of all clock constraints is denoted
$\clockconstr$. That a valuation $v \colon \clocks \to \bbbr_+$ satisfies a
clock constraint~$g$ is defined in a natural way ($v$~satisfies $x \sim c$
whenever $v(x) \sim c$), and we then write $v \models g$. We~denote by~$v_0$
the valuation that assigns zero to all clock variables, by~$v+t$
(with~$t\in\bbbr_+$) the valuation that assigns $v(x)+t$ to all $x\in\clocks$,
and for $R\subseteq\clocks$ we write $[R \leftarrow 0]v$ to denote the
valuation that assigns zero to all variables in~$R$ and agrees with~$v$ for
all variables in~$\clocks\smallsetminus R$.

\begin{definition}
  A \emph{priced timed automaton} (\PPTA for short) is a tuple $\Aut =
  (\Loc,\Init,\penalty-1000\Var,\penalty-1000\Edg,\Inv,(\cost_i)_{1 \leq i
    \leq p})$ where $\Loc$ is a finite set of \emph{locations}, $\Init \in
  \Loc$ is the \emph{initial} location, $\Var$ is a set of \emph{clocks},
  $\Edg \subseteq \Loc\times \Constr[\Var]\times 2^\Var\times\Loc$ is the set
  of \emph{transitions}, $\Inv\colon \Loc \to \Constr[\Var]$ defines the
  \emph{invariants} of each location, and each $\cost_i\colon \Loc\cup\Edg\to
  \Nat$ is a \emph{cost (or price) function}.

  For $S \subseteq \N$, a cost $\cost_i$ is said to be \emph{$S$-sloped} if
  $\cost_i(Q) \subseteq S$. If $S = \{0,1\}$, it is said \emph{stopwatch}. If
  $|S| = n$, we say that the cost $\cost_i$ is $n$-sloped.
\end{definition}
The semantics of a \PPTA~$\Aut$ is given as a labeled timed transition system
$\TTS_{\Aut} = (S, s_0, \trans)$ where $S \subseteq Q \times \bbbr^{\Var}_+$
is the set of states, $s_0=(q_0,v_0)$
is the initial state,
and the transition relation $\mathord{\to} \subseteq S\times(\Edg \cup
\bbbr_+) \times S$ is composed of delay and discrete moves defined as follows:
\begin{enumerate}
\item \textit{(discrete move)} $(q,v) \xrightarrow{e} (q\ip,v\ip)$ if $e =
  (q,g,R,q\ip)\in E$ is s.t. $v \models g$, $v\ip = [R \leftarrow 0]v$, $v\ip
  \models \Inv(q\ip)$. The $i$-th cost of this discrete move is
  $\cost_i\big((q,v) \xrightarrow{e} (q\ip,v\ip)\big) = \cost_i(e)$.
\item \textit{(delay move)} $(q,v) \xrightarrow{t} (q,v+t)$ if $\forall 0 \leq
  t\ip \leq t$, $v+t\ip \models \Inv(q)$. The $i$-th cost of this delay move is
  $\cost_i\big((q,v) \xrightarrow{t} (q,v+t)\big) = t \cdot \cost_i(q)$.
\end{enumerate}
A discrete move or a delay move will be called a \emph{simple move}. A
\emph{mixed move} $(q,v) \xrightarrow{t,e} (q',v')$ corresponds to the
concatenation of a delay move and a discrete move. 
For technical reasons, we only consider non-blocking \PPTA{}s, because
we will further interpret logical formulas over infinite paths.
The $i$-th cost of this mixed move is the sum of the $i$-th costs of the two
moves.

A finite (resp. infinite) \emph{run} 
of a~\PPTA is a finite (resp. infinite) sequence of mixed moves in the
underlying transition system. A run of $\Aut$ will thus be distinguished from
a path in $\TTS_{\Aut}$, which is composed of simple moves and where
stuttering of delay moves is allowed. Note however that a path in
$\TTS_{\Aut}$ is naturally associated with a run in~$\Aut$.
The $i$-th cost of
a run $\varrho$ in $\Aut$ (resp. path $\varrho$ in $\TTS_{\Aut}$) is the sum
of the $i$-th costs of the mixed (resp. simple) moves composing the run (resp.
path), and is denoted $\cost_i(\varrho)$. The length~$|\varrho|$ of a finite
run $\varrho = s_0 \xrightarrow{t_1,e_1} s_1 \xrightarrow{t_2,e_2} \cdots
\xrightarrow{t_n,e_n} s_n$ is~$n$.
A \emph{position} along $\varrho$ is a nonnegative integer~$\pi\leq
|\varrho|$. Given a position~$\pi$, $\varrho[\pi]$ denotes the corresponding
state~$s_{\pi}$, whereas $\varrho_{\leq \pi}$ denotes the finite prefix of
$\varrho$ ending at position~$\pi$, and $\varrho_{\geq \pi}$ is the suffix
starting in~$\pi$. 

\begin{remark}
  In the model of priced timed automata, the cost variables 
only play the role of 
  \emph{observers} (they are \emph{history variables} in the sense
  of~\cite{OG,AL}): the values of these variables don't constrain the
  behaviour of the system (the behaviours of a priced timed automaton are
  those of the underlying timed automaton), but can be used as evaluation
  functions. For~instance, problems such as ``optimal
  reachability''~\cite{BFH+01,ATP01}, ``optimal infinite
  schedules''~\cite{BBL04} or ``optimal reachability timed
  games''~\cite{ABM04,BCFL04,BBR05,BBM06} have recently been investigated. The
  problem we consider in this paper is closely related to these kinds of
  problems: we will use temporal logics as a language for evaluating the
  performances of a system.
\end{remark}

\subsection{Example}
\label{subsec:example}

\input{example}


%% file: example.tex
The \PPTA of Figure~\ref{ex1} models a never-ending process of repairing
problems, which are bound to occur repeatedly with a certain frequency. The
repair of a problem has a certain cost, captured in the model by the cost
variable~$c$. As soon as a problem occurs (modeled by the {\sf Problem}
location) the value of $c$ grows with rate~3, until actual repair is taking
place in one of the locations {\sf Cheap} (rate~2) or {\sf Expensive}
(rate~4). At most $20$ time units after the occurrence of a problem it will
have been repaired one way or another.

\begin{figure}[!ht]
\begin{minipage}{.4\linewidth}
\centering
\begin{tikzpicture}[scale=.7]
   \everymath{\scriptstyle}
  \draw (0,0) node[draw,rounded corners=1mm,line width=.7pt] (A) {$\begin{array}{c} \dot{c}=0 \\ x \leq 9 \end{array}$};
  \draw (3,0) node[draw,rounded corners=1mm,line width=.7pt] (B) {$\begin{array}{c} \dot{c}=3 \\ x \leq 10 \end{array}$};
  \draw (6,1.5) node[draw,rounded corners=1mm,line width=.7pt] (C) {$\begin{array}{c} \dot{c}=2 \\ x < 20 \end{array}$};
  \draw (6,-1.5) node[draw,rounded corners=1mm,line width=.7pt] (D) {$\begin{array}{c} \dot{c}=4 \\ x \leq 15 \end{array}$};
  \draw[-latex'] (A)--(B);
  \draw[-latex'] (B).. controls +(.5,1) and +(-2,0) ..(C) node[near end,above,sloped] {$x \geq 2$};
  \draw[-latex'] (B).. controls +(.5,-1) and +(-2,0) ..(D) node[near end,below,sloped] {$x \geq 4$};
  \draw[-latex'] (C).. controls +(0,2) and +(0,2) ..(A) node[midway,sloped,above] {$x=20$, $x:=0$, $c+=5$};
  \draw[-latex'] (D).. controls +(0,-2) and +(0,-2) ..(A) node[midway,sloped,below] {$x=15$, $x:=0$};
  \draw (2.4,1.1) node (E) {\sffamily\footnotesize Problem};
  \draw (6,.35) node (F) {\sffamily\footnotesize Cheap};
  \draw (6,-.4) node (G) {\sffamily\footnotesize Expensive};
  \draw (1.3,-.65) node (H) {\sffamily\footnotesize OK};
\end{tikzpicture}
\caption{Repair problem as a \PPTA}\label{ex1}
\end{minipage}
\hfill
\begin{minipage}{.52\linewidth}
\centering
            \begin{tikzpicture}[xscale=.25,yscale=.35]
              \everymath{\scriptstyle}
              \draw[latex'-latex'] (0,12)--(0,0)--(12,0);
              \foreach \x in {2,4,6,8,10} \draw (\x,0)--(\x,-.4);
              \foreach \x in {2,4,6,8,10} \draw (\x,-1) node {$\x$};
              \draw (12,-1) node {$x$};
              \foreach \x in {2,4,6,8,10} \draw (0,\x)--(-.4,\x);
              \foreach \x in {10,20,30,40,50} \draw[yscale=.2] (-1.2,\x) node {$\x$};
              \draw (-1.2,12) node {$c$};
              \draw[yscale=.2,line width=.8pt] (0,47)--(2,41) node[near start] {{\tiny \phantom{tototototototototo} \textcolor{black}{Wait in} {Problem}}};
              \draw[yscale=.2,line width=.1mm] (2,41) -- +(0:2mm) -- +(0:0mm) -- +(90:10mm) -- +(0:0mm) -- +(180:2mm)  -- +(0:0mm) -- +(270:10mm);
              \draw[yscale=.2,line width=.8pt] (2,41)--(5,35) node[near start] {{\tiny \phantom{totototototototo} \textcolor{black}{Goto} {Cheap}}};
              \draw[yscale=.2,line width=.1mm] (5,35) -- +(0:2mm) -- +(0:0mm) -- +(90:10mm) -- +(0:0mm) -- +(180:2mm)  -- +(0:0mm) -- +(270:10mm);
              \draw[yscale=.2,line width=.8pt] (5,35)--(10,20) node[near start] {{\tiny \phantom{tototototototototo} \textcolor{black}{Wait in} {Problem}}};
              \fill[yscale=.2] (10,20) ellipse (5pt and 25pt);
              \draw[yscale=.2] (10,18) node {{\tiny Goto {Expensive}}};
              \draw[yscale=.2] (10,18) node {{\tiny \phantom{totototototototototototototototototoz}}};
            \end{tikzpicture}
\caption{Minimum cost of repair and associated strategy in location {\sf Problem}}\label{ex2}
\end{minipage}
\end{figure}

In~this setting we are interested in properties concerning the cost of
repairs. For instance, we would like to express that whenever a problem
occurs, it \emph{may} be repaired (\textit{i.e.}~reach the
location~$\textsf{OK}$) within a total cost of~$47$. In fact Figure~\ref{ex2}
gives the minimum cost of repair ---as~well as an optimal strategy--- for any
state of the form $(\textsf{Problem},x)$ with~$x\in[0,10]$. Correspondingly,
the minimum cost of reaching~$\textsf{OK}$ from states of the form
$(\textsf{Cheap},x)$ (resp. $(\textsf{Expensive},x)$) is given by the
expression~$45-2x$ (resp.~$60-4x$). Symmetrically, we would like to express
properties on the worst cost to repair, or to link the uptime with the (best,
worst) cost of repairing. As~will be illustrated later, extending temporal
logics with cost informations provides a nice setting for expressing such
properties.

%% file: modelchecking.tex

\section{Model Checking Branching-Time Logics}\label{sec:WCTL}

We first focus on the case of branching-time logics. From this point on, \AP
denotes a fixed, finite, non-empty set of atomic propositions. We first define
the cost-extended version of~\CTL.

\subsection{The Logic \texorpdfstring{\WCTL}{WCTL}}


The logic \WCTL{}\footnote{\WCTL stands for ``Weighted \CTL'',
  following~\cite{BBR04} terminology. It would have been more natural
  to call it ``Priced \CTL'' (\PCTL) in our setting, but this would
  have been confusing with ``Probabilistic
  \CTL''~\cite{HJ-fac94}.}~\cite{BBR04} extends \CTL with cost
constraints.  Its syntax is given by the following grammar:
\[
\WCTL\ni \phi\ ::=\ a\ \mid\ \neg \phi\ \mid\ \phi\ou\phi\ \mid\
\E\phi\U[\cost\sim c]\phi\ \mid\ \A\phi\U[\cost\sim c]\phi
\]
where $a \in \AP$, $\cost$ is a cost function, $c$ ranges over~\Nat, and
$\mathord\sim\in\{ \mathord<,\mathord\leq,\penalty-1000\mathord=,
\penalty-1000\mathord\geq,\mathord>\}$.

We interpret formulas of \WCTL over labeled \PPTA, \textit{i.e.} \PPTA having
a labeling function~$\Lab$ which associates with every location~$q$ a subset
of~$\AP$. We~identify each cost appearing in the \WCTL formulas with the cost
having the same name in the model (which is assumed to exist).

\begin{definition}
  Let \Aut{} be a labeled \PPTA. The satisfaction relation of \WCTL is defined
  over configurations $(q,v)$ of \Aut as follows:\par
  \[\begin{array}{r@{\ \ }c@{\ \ }l}
    (q,v)\models a &\Leftrightarrow & a\in\Lab(q) \\
    (q,v)\models \non\phi &\Leftrightarrow & (q,v)\not\models \phi \\
    (q,v)\models\phi\ou\psi &\Leftrightarrow & (q,v)\models\phi \text{ or }
    (q,v)\models\psi \\
    (q,v)\models \E\phi\U[\cost \sim c]\psi &\Leftrightarrow & \text{there is an
      infinite run }\varrho 
    \text{ in }\Aut \\
    & & \text{from }(q,v)\text{ s.t. }\varrho\models\phi\U[\cost \sim c]\psi \\
    (q,v)\models \A\phi\U[\cost \sim c]\psi &\Leftrightarrow &\text{any infinite run
    }\varrho 
    \text{ in }\Aut \text{ from } (q,v)\\
    & & \text{satisfies }
    \varrho\models\phi\U[\cost \sim c]\psi \\
    \varrho \models \phi \U[\cost \sim c] \psi &
    \Leftrightarrow & \text{there exists a position }\pi>0 \text{  
      along}\ \varrho\ \text{s.t.}\\ 
    & & \varrho[\pi] \models \psi,\ \text{for every position}\ 0< \pi' < \pi, \\
    & & 
    \varrho[\pi'] \models \phi,\ \text{and}\ \cost(\varrho_{\leq \pi})\sim c
  \end{array}\]
%
%
  If \Aut is not clear from the context, we may write $\Aut,(q,v) \models 
  \phi$ instead of simply $(q,v) \models \phi$.
\end{definition}

As usual, we will use shorthands such as ``$\texttt{true} \eqdef a\ou \non
a$'', ``$(\varphi \Rightarrow \psi) \eqdef \neg \varphi \vee \psi$'', ``$\E
\F[\cost\sim c] \phi \eqdef \E \texttt{true} \,\U[\cost\sim c] \phi$'', and
``$\A \G[\cost\sim c] \phi \eqdef \neg \E \F[\cost\sim c] \neg \phi$''.
Moreover, if the cost function~$\cost$ is unique or clear from the context, we
may write $\phi\U[\sim c]\psi$ instead of $\phi\U[\cost \sim c]\psi$.
Finally, we omit to mention the subscript~``$\sim c$'' when it is equivalent 
to ``$\geq 0$'' (thus imposing no real constraint). 

\begin{example}
  We go back to our example of Section~\ref{subsec:example}. That it is always
  possible to repair a problem with cost at most~$47$ can be expressed
  in~\WCTL with the following formula:
  \[
  \A\G\big(\textsf{Problem}\Rightarrow \E\F[c \leq 47] \textsf{OK}\big).
  \]
  We can also express that the worst cost to repair is~$56$, in the sense that
  state~\textsf{Repair} can always be reached within this cost:
  \[
  \A\G\big(\textsf{Problem}\Rightarrow \A\F[c \leq 56] \textsf{OK}\big).
  \]

  Now, considering time as a special case of a cost (with constant slope~$1$),
  we can express properties relating the time elapsed in the \textsf{OK} state
  and the cost to repair:
  \[
  \A\G\big(\neg \E (\textsf{OK}\; \U[t \geq 8] (\textsf{Problem} \wedge \neg
  \E\F_{c<30} \textsf{OK}))\big).
  \]
  This expresses that if the system spends at least $8$ (consecutive) time
  units in the \textsf{OK} state, then the next \textsf{Problem} can be
  repaired with cost at most~$30$.
\end{example}

\medskip

\bigskip The main result of this section is the following theorem:

\begin{theorem}
  \label{theo:full}
  Model-checking \WCTL on one-clock \PPTA is \PSPACE-Complete.
\end{theorem}

The \PSPACE lower bound can be proved by a direct adaptation of
the \PSPACE-Hardness proof for the model-checking of \TCTL, the restriction of
\WCTL to time constraints, over one-clock timed automata~\cite{LMS04}.

The \PSPACE upper bound is more involved, and will be done in two steps:
\begin{enumerate}
\item first we will exhibit a set of regions which will be correct for
  model-checking \WCTL formulas, see Section~\ref{section:granularite};
\item then we will use this result to propose a \PSPACE algorithm for
  model-checking \WCTL, see Section~\ref{sec:algo}.
\end{enumerate}

Finally, it is worth reminding here that the model-checking of \WCTL over
priced timed automata with three clocks is undecidable~\cite{BBM06}.

\subsection{Sufficient Granularity for \texorpdfstring{\WCTL}{WCTL}}

\input{WCTL-new}


%% file: WCTL-new.tex
\label{section:granularite}

The proof of Theorem~\ref{theo:full} partly relies on the following
proposition, which exhibits, for every \WCTL formula $\Phi$, a set of
\emph{regions} within which the truth of $\Phi$ is uniform. Note that
these are not the classical regions as defined in~\cite{AD94,ACD93},
because their granularity needs to be refined in order to be
correct. Computing a sufficient granularity was already a key step for
checking duration properties in simple timed
automata~\cite{BES-lics93}.

\begin{proposition}\label{mainprop}
  Let $\Phi$ be a \WCTL formula and let~\Aut be a one-clock \PPTA. Then there
  exist a finite set of constants $\{a_0,...,a_n\}$
  satisfying the following conditions:  
  \begin{enumerate}[$\bullet$]
  \item $0 = a_0 < a_1 < \ldots < a_n < a_{n+1} = +\infty$;
  \item for every location~$q$ of~\Aut, for every $0 \leq i \leq n$,
  the truth of~$\Phi$ is uniform over $\{(q,x) \mid a_i < x < a_{i+1}\}$;
  %
  \item $\{a_0,...,a_n\}$ contains all the constants appearing in clock
    constraints of~\Aut;
  \item the constants are integral multiples of~$1/C^{\height\Phi}$
    where $\height\Phi$ is the \emph{constrained temporal height}
    of~$\Phi$, \textit{i.e.}, the maximal number of nested constrained
    modalities\footnote{With "constrained modality" we mean a modality
      decorated with a constraining interval different from
      $(0,+\infty)$.}  in~$\Phi$, and~$C$ is the lcm of all positive
    costs labeling a location of~\Aut;
  \item $a_n$ equals 
    the largest constant~$M$ appearing in the guards of~\Aut;
  \end{enumerate}
In particular, we have $n \leq M \cdot C^{\height\Phi}+1$.
\end{proposition}

As a corollary, we recover the partial decidability result of~\cite{BBR04},
stating that the model-checking of one-clock \PPTA with a \emph{stopwatch
  cost}\footnote{\textit{I.e.}, cost with rates in~$\{0,1\}$.} against \WCTL
formulas is decidable using classical one-dimensional regions of timed
automata (\textit{i.e.},~with granularity~$1$).

\begin{proof}
  The proof of this proposition is by structural induction on~$\Phi$. The
  cases of atomic propositions and boolean combinations are straightforward;
  unconstrained modalities require no refinement of the granularity (the basic
  \CTL algorithm is correct and does not need to refine the granularity); we will thus focus on constrained
  modalities.

  \subsubsection{We first assume that~\Aut{} has no discrete costs}
  (\textit{i.e.}~$\cost(T)=\{0\}$), the extension to the general case will be
  presented at the end of the proof.

  \bigskip 
  \noindent$\blacktriangleright$ \ \textbf{We first focus on the case when $\Phi=\E\phi\U[\cost \sim
    c]\psi$} (we simply write $\Phi=\E\phi\U[\sim c]\psi$, and assume
  that~$\cost$ is the only cost of~$\Aut$, as its other costs play no role in
  the problem).
  Assume that the result has been proved for the \WCTL subformulas $\phi$
  and~$\psi$, and that we have merged all constants for $\phi$ and~$\psi$: we
  thus have constants $0 = a_0 < a_1 < \ldots < a_n < a_{n+1} = +\infty$ such
  that for every location~$q$ of~\Aut, for every $0 \leq i \leq n$, the truth
  of~$\phi$ and that of~$\psi$ are both uniform over $\{(q,x) \mid a_i < x <
  a_{i+1}\}$. By induction hypothesis, the granularity of these constants is
  $1/C^{\max(\height\phi,\height\psi)} = 1/C^{\height\Phi-1}$. We will exhibit
  extra constants such that the above proposition then also holds for the
  formula $\Phi$. For the sake of simplicity, we will call \emph{regions} all
  elementary intervals $(a_i,a_{i+1})$ and singletons $\{a_i\}$. \medskip

  In order to compute the set of states satisfying~$\E\phi \U[\sim c] \psi$,
  for every state~$(q,x)$ we compute all costs of paths from~$(q,x)$ to some
  region~$(q',r)$, along which~$\phi$ always
  holds after a discrete action has been done, and such that a $\psi$-state
  can immediately be reached via a discrete action from~$(q',r)$.
  We then check whether we can achieve a cost satisfying~``$\sim c$'' for the
  mentioned $\psi$-state. We thus first explain how we compute the set of
  possible costs between a state~$(q,x)$ and a region~$(q',r)$ in~\Aut.
  Indeed, for checking the existence of a run satisfying $\phi \U[\sim c]
  \psi$, we will first remove discrete transitions leading to states not
  satisfying $\phi$, and then compute all possible costs of runs from $(q,x)$
  to some $(q',r)$, where $(q',r)$ is a $\psi$-state just reached by a
  discrete action, in the restricted graph. 

  For each index~$i$, we restrict the automaton~\Aut to transitions
  whose guards contain the interval~$(a_i,a_{i+1})$, and that do not
  reset the clock. We denote by~$\Aut_i$ this restricted
  automaton. Let~$q$ and~$q'$ be two locations of~$\Aut_i$. As~stated
  by the following lemma, the set of costs of paths between $(q,a_i)$
  and~$(q',a_{i+1})$ is an interval that can be easily computed:
  \begin{lemma}\label{lemma4.2}
    We assume~$a_{i+1} \neq +\infty$. Let $S_i(q,q')$ be the set of locations
    that are reachable from~$(q,a_i)$ and co-reachable from~$(q',a_{i+1})$
    in~$\TTS_{\Aut_i}$, and assume it is non-empty (\emph{i.e.}, there is a
    path joining those two states). Let $c^{i,q,q'}_{\min}$
    and $c^{i,q,q'}_{\vphantom{\min}\max}$ be the minimum and maximum costs
    among the costs of locations in~$S_i(q,q')$.
    Then the set of all possible costs of paths in $\TTS_{\Aut_i}$ going
    from~$(q,a_i)$ to~$(q',a_{i+1})$ is an interval $\langle
    (a_{i+1}-a_i)\cdot c^{i,q,q'}_{\min}, (a_{i+1}-a_i)\cdot
    c^{i,q,q'}_{\vphantom{\min}\max}\rangle$. The interval is left-closed iff
    there exist two locations $r$ and~$s$ (with possibly~$r=s$) in~$S_i(q,q')$
    with cost~$c^{i,q,q'}_{\min}$ such that\footnotemark{} $(q,a_i)
    \leadsto^*_{\Aut_i} (r,a_i)$, $(r,a_i) \leadsto^*_{\Aut_i} (s,a_{i+1})$,
    and $(s,a_{i+1}) \leadsto^*_{\Aut_i} (q',a_{i+1})$. The interval is
    right-closed iff there exists two locations $r$ and $s$ in $S_i(q,q')$
    with cost $c^{i,q,q'}_{\vphantom{\min}\max}$ such that $(q,a_i)
    \leadsto^*_{\Aut_i} (r,a_i)$, $(r,a_i) \leadsto^*_{\Aut_i} (s,a_{i+1})$,
    and $(s,a_{i+1}) \leadsto^*_{\Aut_i} (q',a_{i+1})$.
  \end{lemma}
  \footnotetext{The notation $\alpha \leadsto^*_{\Aut_i} \alpha'$ means that
    there is a path in~$\TTS_{\Aut_i}$ from $\alpha$ to~$\alpha'$.}

  The conditions on left\slash right-closures characterize the fact
  that it is possible to instantaneously reach\slash leave a location
  with minimal\slash maximal cost, or if a small positive delay has to
  elapse (due to a strict guard).

  \begin{proof}
    Obviously the costs of all paths in $\TTS_{\Aut_i}$ from $(q,a_i)$ to
    $(q',a_{i+1})$ belong to the interval $(a_{i+1}-a_i)\cdot
    [c^{i,q,q'}_{\min}, c^{i,q,q'}_{\vphantom{\min}\max}]$. We will now prove
    that the set of costs is an interval containing $(a_{i+1}-a_i)\cdot
    (c^{i,q,q'}_{\min}, c^{i,q,q'}_{\vphantom{\min}\max})$.

    \begin{figure}[!ht]
      \begin{center}
        \begin{tikzpicture}
          \draw[dashed,rounded corners=4mm] (2,1) --
          (2.2,1.9) node[midway,coordinate] (a2) {} -- (3.7,2.3) -- (5,2.5) --
          (5.7,1.9) -- (6,1.1) node[coordinate] (a5) {}
          node[pos=.3,coordinate] (a3) {} -- (5.8,.1)  
          node[pos=.6,coordinate] (a4) {} --
          (4.3,-.4) -- (2.2,-.4) -- (2,1) node[midway,coordinate] (a1) {};
          \path (0,1) node[draw=black,circle,fill=white!70!black] (A) 
          {\hbox to 20pt{\hss$\scriptstyle q,a_i$\hss}};
          \path (8,1) node[draw=black,circle,fill=white!70!black] (B) 
          {\hbox to 20pt{\hss$\scriptstyle q',a_{i+1}$\hss}};
          \path (3.7,1.8) node[draw=black,circle] (C) {};
          \path (5,.6) node[draw=black,circle,fill=white!20!black] (D) {};
          \draw (7,2.4) circle(1.5mm) node[right=3mm]
            {\hbox to 60pt{\small maximal cost\hfil}};
          \draw[fill=white!20!black] (7,2) circle(1.5mm) node[right=3mm]
            {\hbox to 60pt{\small minimal cost\hfil}};
          \path (3,.9) node[draw,circle,fill=white!70!black] (E) {};
          \path (3.6,.3) node[draw,circle,fill=white!70!black] (F) {};
          \path (4.6,1.6) node[draw,circle,fill=white!70!black] (G) {};
          \draw[-latex'] (a2) -- (C);
          \draw[-latex'] (C) -- (G); \draw[-latex'] (G) -- (a5);
          \draw[-latex'] (a1) -- (E); \draw[-latex'] (E) -- (D);
          \draw[-latex'] (D) -- (a4);
          \draw[-latex'] (A) -- (a1);
          \draw[-latex'] (A) -- (a2);
          \draw[dashed,-latex'] (a1) -- +(-10:7mm);
          \draw[dashed,-latex'] (a1) -- +(10:7mm);
          \draw[dashed,-latex'] (a2) -- +(-10:7mm);
          \draw[-latex'] (a3) -- (B);
          \draw[-latex'] (a4) -- (B);
          \draw[-latex'] (a5) -- (B);
          \draw[latex'-,dashed] (a3) -- +(-170:7mm);
          \draw[latex'-,dashed] (a3) -- +(-190:7mm);
          \draw[latex'-,dashed] (a4) -- +(-170:7mm);
          \draw[latex'-,dashed] (a5) -- +(-170:7mm);
        \end{tikzpicture}
      \end{center}
      \caption{The set of costs between two states is an
        interval.}\label{fig-interval}
    \end{figure}
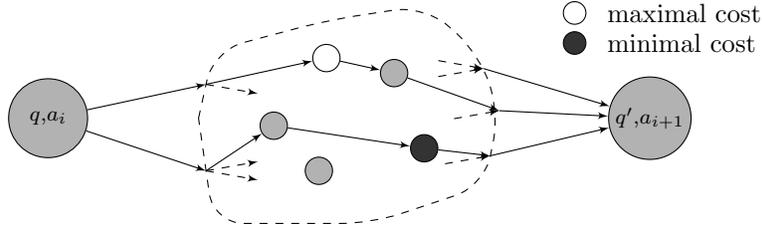
    Let $\tau_{\min}$ (resp.~$\tau_{\vphantom{\min}\max}$) be a sequence of
    transitions in $\Aut_i$ leading from $(q,a_i)$ to~$(q',a_{i+1})$ and going
    through a location with minimal (resp.~maximal) cost
    (see~Figure~\ref{fig-interval}). Easily enough, the possible costs of the
    paths following~$\tau_{\min}$ (resp.~$\tau_{\vphantom{\min}\max}$) form an
    interval whose left (resp.~right) bound is $c^{i,q,q'}_{\min} \cdot
    (a_{i+1}-a_i)$ (resp.~$c^{i,q,q'}_{\vphantom{\min}\max} \cdot
    (a_{i+1}-a_i)$).

    Now, if $c$ and $c'$ are the respective costs of $q$ and~$q'$, then
    $\frac{1}{2} \cdot (c+c') \cdot (a_{i+1}-a_i)$ is in both intervals.
    Indeed, the path following~$\tau_{\min}$
    (resp.~$\tau_{\vphantom{\min}\max}$) which delays
    $\frac{1}{2}\cdot(a_{i+1}-a_i)$ time units in~$q$, then directly goes
    to~$q'$ and waits there for the remaining $\frac{1}{2}\cdot(a_{i+1}-a_i)$
    time units achieves the above-mentioned cost. This implies that the set of
    all possible costs is an interval.

    The bound $c^{i,q,q'}_{\min} \cdot (a_{i+1}-a_i)$ is reached iff there is
    a path from $(q,a_i)$ to~$(q',a_{i+1})$ which delays only in locations
    with cost~$c^{i,q,q'}_{\min}$. This is precisely the condition expressed
    in the lemma. The same holds for the upper bound
    $c^{i,q,q'}_{\vphantom{\min}\max} \cdot (a_{i+1}-a_i)$.
  \end{proof}

  Similar results clearly hold for other kinds of regions:
  \begin{enumerate}[$\bullet$]
  \item between a state~$(q,a_i)$ and a region~$(q',(a_i,a_{i+1}))$
    with~$a_{i+1}\not=+\infty$, the set of possible costs is an
    interval~$\langle 0,c^{i,q,q'}_{\vphantom{\min}\max}\cdot (a_{i+1}-a_i))$,
    where~$0$ can be reached iff it is possible to go from~$(q,a_i)$ to some
    state~$(q'',a_i)$ co-reachable from $(q',x)$ for some $x \in
    (a_i,a_{i+1})$, and~$\cost(q'')=0$.
  \item between a state~$(q,x)$, with~$x\in (a_i,a_{i+1})$,
    and~$(q',a_{i+1})$, the set of costs is~$(a_{i+1}-x) \cdot \langle
    c^{i,q,q'}_{\min}, c^{i,q,q'}_{\vphantom{\min}\max} \rangle$, with similar
    conditions as above for the bounds of the interval.
  \item between a state~$(q,x)$, with~$x\in (a_i,a_{i+1})$, and
    region~$(q',(a_i,a_{i+1}))$ (assuming~$a_{i+1}\not=+\infty$), the set of
    possible costs is $[0,c^{i,q,q'}_{\vphantom{\min}\max}\cdot (a_{i+1}-x))$;
  \item between a state~$(q,a_n)$ and a region~$(q',(a_n,+\infty))$,
    the set of possible costs is either~$[0,0]$, if
    no positive cost rate is reachable and co-reachable, or~$\langle
    0,+\infty)$ otherwise. If the latter case, $0$~can be achieved iff it is
    possible to reach a state~$(q'',a_n)$ with~$\cost(q'')=0$;
  \item between a state~$(q,x)$ with~$x\in (a_n,+\infty)$
    and a region~$(q',(a_n,+\infty))$, the set of costs
    is either~$[0,0]$ or~$[0,+\infty)$, with the same conditions as
    previously.
  \end{enumerate}

  We use these computations and build a graph~$G$ 
  labeled by intervals which will store all possible costs between symbolic
  states (\textit{i.e.},~pairs~$(q,r)$, where $q$ is a location and $r$ a
  region) in~$\TTS_{\Aut}$. Vertices of~$G$ are pairs~$(q,\{a_i\})$ and
  $(q,(a_i,a_{i+1}))$, and tuples $(q,x,\{a_i\})$ and $(q,x,(a_i,a_{i+1}))$,
  where $q$ is a location of $\Aut$.  Their roles are as follows:
  vertices of the form~$(q,x,r)$ are used to initiate a computation,
  they represent a state~$(q,x)$ with~$x\in r$.  States~$(q,\{a_i\})$
  are ``regular'' steps in the computation, while
  states~$(q,(a_i,a_{i+1}))$ are used either for finishing a
  computation, or just before resetting the clock (there will be no
  edge from~$(q,(a_i,a_{i+1}))$ to any~$(q',\{a_{i+1}\})$).

  Edges~of~$G$~are defined as follows:\label{graph_construction}
  \begin{enumerate}[$\bullet$]
  \item $(q,\{a_i\}) \rightarrow (q',\{a_{i+1}\})$ if there is a path from
    $(q,a_i)$ to~$(q',a_{i+1})$. This edge is then labeled with an interval
    $\langle (a_{i+1}-a_i)\cdot c^{i,q,q'}_{\min}, (a_{i+1}-a_i) \cdot
    c^{i,q,q'}_{\vphantom{\min}\max}\rangle$, the nature of the interval
    (left-closed and\slash or right-closed) depending on the criteria exposed
    in Lemma~\ref{lemma4.2}.
  \item $(q,\{a_i\}) \rightarrow (q',\{a_i\})$ if there is an instantaneous
    path from $(q,a_i)$ to $(q',a_i)$ in \Aut, the edge is then labeled with
    the interval $[0,0]$ (because we assumed there are no discrete costs
    on transitions of~$\Aut$).
  \item $(q,\{a_i\}) \rightarrow (q',\{a_0\})$ if there is a transition
    in~$\Aut$ enabled when the value of the clock is~$a_i$ and resetting the
    clock. It is labeled with~$[0,0]$.
  \item $(q,(a_i,a_{i+1})) \rightarrow (q',\{a_0\})$ if there is a transition
    in~$\Aut$ enabled when the value of the clock is in $(a_i,a_{i+1})$ and
    resetting the clock. It~is labeled with~$[0,0]$.
  \item $(q,\{a_i\}) \rightarrow (q',(a_i,a_{i+1}))$ if there is a path
    from~$(q,a_i)$ to some $(q',\alpha)$ with $a_i < \alpha < a_{i+1}$. This
    edge is labeled with the interval $\langle 0, (a_{i+1}-a_i)\cdot
    c^{i,q,q'}_{\vphantom{\min}\max})$.
  \item $(q,x,\{a_i\}) \rightarrow (q,\{a_i\})$ labeled with~$[0,0]$.
  \item $(q,x,(a_i,a_{i+1})) \rightarrow (q',\{a_{i+1}\})$ if there is a path
    from some~$(q,\alpha)$ with~$a_i < \alpha < a_{i+1}$ to~$(q',a_{i+1})$.
    This edge is labeled with $(a_{i+1}-x)\cdot \langle c^{i,q,q'}_{\min},
    c^{i,q,q'}_{\vphantom{\min}\max}\rangle$.
  \item $(q,x,(a_i,a_{i+1})) \rightarrow (q',(a_i,a_{i+1}))$ labeled
    with $[0, (a_{i+1}-x)\cdot c^{i,q,q'}_{\vphantom{\min}\max})$.
  \end{enumerate}

  Figure~\ref{fig-graph}
  represents one part of this graph. Note that each path~$\pi$ of this
  graph is naturally associated with an interval~$\iota(\pi)$
  (possibly depending on variable~$x$ if we start from a node
  $(q,x,(a_i,a_{i+1}))$) by summing up all intervals labeling
  transitions of $\pi$. 
\input{fig-graph}

  The correctness of graph $G$ w.r.t.~costs is stated by the following lemma,
  which is a direct consequence of the previous investigations.

  \begin{lemma}\label{lemma4.5}
    Let $q$ and $q'$ be two locations of $\Aut$. Let $r$ and $r'$ be
    two regions, and let $\alpha \in r$.
    Let~$d\in\R^+$. There exists a path $\pi$ in~$G$ from a
    state~$(q,x,r)$ to~$(q',r')$ with  $\iota(\pi)(\alpha)\ni d$
    if, and only if, there is a path in~$\TTS_{\Aut}$ with total cost~$d$, and
    going from~$(q,\alpha)$ to some~$(q',\beta)$ with~$\beta\in r'$.
  \end{lemma}

  \begin{corollary}\label{coro-8}
    Fix two regions $r$ and~$r'$. Then the set of possible costs of
    paths in~$G$ from $(q,x,r)$ to~$(q',r')$ is of the form
    \[
    \bigcup_{m\in\Nat} \langle \alpha_m - \beta_m\cdot x, \alpha'_m -
    \beta'_m\cdot x\rangle
    \]
    (possibly with $\beta_m$ and\slash or $\beta'_m=0$, and\slash or
    $\alpha'_m=+\infty$).
    Moreover, 
    \begin{enumerate}[$\bullet$]
    \item all constants~$\alpha_m$ and~$\alpha'_m$ are either integral
      multiples of~$1/C^{\max(\height\phi,\height\psi)}$ or~$+\infty$,
      and constants~$\beta_m$ and~$\beta'_m$ are either costs of the
      automaton or~$0$;
    \item if~$r=(a_n,+\infty)$, then $\beta_m=\beta'_m=0$ for all~$m$.
    \end{enumerate}
  \end{corollary}

  \begin{proof}
    Applying Lemma~\ref{lemma4.5}, the union of the costs of all paths in~$G$
    from $(q,x,r)$ to $(q',r')$ represents the set of all possible costs of
    paths in~$\TTS_{\Aut}$ from~$(q,\alpha)$ with $\alpha \in r$ to
    some~$(q',\beta)$ with~$\beta \in r'$.
    This set can be written as the countable union, for each~$m\in \Nat$, of
    the costs of paths of length~$m$ in~$G$, thus a countable union of (a
    finite union of) intervals. Now, any path in~$G$ contains at most one
    transition issued from a state~$(q,x,r)$. Thus, coefficients~$\beta_m$ are
    either~$0$, or the cost of some location of~$\Aut$.

    Coefficients~$\alpha_m$ are then integral combinations of terms of the
    form $c \cdot (a_{i+1}-a_i)$ where $c$~is the cost of some location.
    As~all~$a_i$'s are integral multiples of
    $1/C^{\max(\height\phi,\height\psi)}$, we get what we expected. 
    The~special form for the unbounded region is obvious from the construction of~$G$.
  \end{proof}

  \begin{lemma}\label{lemma:gran}
    For every location $q$, and for~$\Phi = \E \phi \U[\sim c] \psi \in\WCTL$,
    the set of clock values~$x$ such that $(q,x)$ satisfies $\Phi$ is a finite
    union of intervals. Moreover,
    \begin{enumerate}[$\bullet$]
    \item the bounds of those intervals are integral multiples
      of~$1/C^{\height\Phi}$; 
    \item the largest finite bound of those intervals is at most the
      maximal constant appearing in the guards of the automaton.
    \end{enumerate}
  \end{lemma}

  \begin{proof}
    The set of clock values~$x$ such that $(q,x)$ satisfies $\E \phi \U[\sim
    c] \psi$ can be written as
    \[
    \bigcup_{r \text{ region}} \{x \in r\ \mid\ (q,x)\models \E \phi \U[\sim
    c] \psi\}.
    \]
    There is a finite number of regions. For the unbounded region, the set of
    possible costs does not depend on the initial value of~$x$, and thus
    either the whole region satisfies the formula, or no point in that region
    does. Fix a bounded region $r$, and $x \in r$. Then, $(q,x) \models \E
    \phi \U \psi$ if, and only if there exists a path in $\TTS_{\Aut}$ from
    $(q,x)$ to some $(q',r')$ such that $(i)$ a $\psi$-state is immediately
    reachable from $(q',r')$ by a discrete move, and $(ii)$ along that path,
    all states traversed just after a discrete move satisfy $\phi$. 
    For each pair~$(q,r)$ leading to a~$\psi$-state, we can applying
    Corollary~\ref{coro-8} on the graph just obtained after having removed
    discrete transitions not leading to a $\phi$-state.
    The set of
    possible costs of paths satisfying $\phi \U \psi$ is then a (countable) union of the form
    $\bigcup_{m\in\Nat} \langle \alpha_m - \beta_m\cdot x, \alpha'_m -
    \beta'_m\cdot x\rangle$ with the constraints on constants described in the
    previous corollary. We assume that $r = (a_i,a_{i+1})$ 
    and that the constraint $\sim
    c$ is either $\leq c$, or $<c$, or $=c$ (the other cases would be
    handled in a similar way). If $\alpha_m - \beta_m \cdot a_i
    > c$, then the interval $\langle \alpha_m - \beta_m\cdot x, \alpha'_m -
    \beta'_m\cdot x\rangle$ plays no role for the satisfaction of formula $\E
    \phi \U[\sim c] \psi$ in the region $r$, we can thus remove this interval
    from the union. Now, $\beta_m$ is an integer which is either null or
    divides $C$. Thus as $\alpha_m$ is an integral multiple of
    $1/C^{\max(\height\phi,\height\psi)} = 1/C^{\height\Phi-1}$, left-most
    bounds of interesting intervals can be $\alpha_m-\beta_m\cdot x$ for
    finitely many $\alpha_m$'s and $\beta_m$'s (with the further option closed
    or open). Fix some $\alpha$ and $\beta$, and also fix some $\beta'$. For
    two intervals $\langle \alpha - \beta\cdot x, \alpha'_1 - \beta'\cdot
    x\rangle$ and $\langle \alpha - \beta\cdot x, \alpha'_2 - \beta'\cdot
    x\rangle$ in the above union, it is sufficient to keep only the one with
    the largest $\alpha'$ (because the other is included in this interval).
    Thus, in the above countable union of intervals, we can select a finite
    union of intervals which will be sufficient for checking property $\E \phi
    \U[\sim c] \psi$ in region~$r$. 

    We thus assume that the set of costs of paths which may witness formula
    $\phi \U[\sim c] \psi$ is a finite union $\bigcup_{m = 1}^k \langle
    \alpha_m - \beta_m\cdot x, \alpha'_m - \beta'_m\cdot x\rangle$ with
    $\alpha_m$ and~$\alpha'_m$ in~$\Nat/C^{\height\Phi-1}$ and~$\beta_m$
    and~$\beta'_m$ in~$(C/\Nat^* \cap \Nat) \cup \{0\}$.
    Now, the bounds~$a'_i$ of the intervals of positions where~$\Phi$ holds
    should correspond to values of~$x$ where one of the
    bounds~$\alpha_m-\beta_m\cdot x$ or $\alpha'_m-\beta'_m\cdot x$ exactly
    equals~$c$. It easily follows that those bounds~$a'_i$ are integral
    multiples of~$1/C^{\height\Phi}$, as required.

    This proves that we get only finitely many new intervals, and that the
    largest constant is the same as for~$\phi$ and~$\psi$ (because of the
    initial remark on the unbounded region), thus it is the largest constant
    appearing in the automaton.
  \end{proof}
  
  This concludes the induction step for formula $\E\phi\U[\sim c]\psi$ when
  the automaton has no discrete cost. We will now handle the cases of the
  formulas $\E\G[\geq c] \texttt{false}$ and $\E\G[=c] \texttt{false}$ before
  giving several equivalences to handle all the other cases.

  \bigskip 
\noindent$\blacktriangleright$ \ \textbf{We now consider the formulas $\Phi = \E\G[=c]
    \texttt{false}$ and $\Phi = \E\G[\geq c] \texttt{false}$:} 
    handling those modalities is sufficient for our proof, as we explain later.

    To handle those two
  formulas, we will extend the graph~$G$ defined previously for the initial
  automaton (with non-refined classical regions). We add to the graph $G$ new
  ``final'' states which are triples $\overline{(q,y,r)}$ (we overline it to
  distinguish it from the initial states). Such a state has the same incoming
  transitions as the state~$(q,r)$, except that we will enforce the final
  value of the clock be $y$, and not any value in $r$. For instance, a
  transition $(q,\{a_i\}) \rightarrow \overline{(q',y,(a_i,a_{i+1}))}$ will be
  labeled by the interval $\langle 0, (y-a_i) \cdot c_{\max}^{i,q,q'}]$
  (remember the construction of the graph on
  page~\pageref{graph_construction}). From each of these new final states, 
  we~add an outgoing transition labeled by a finite union of intervals
  corresponding to all the costs of a single mixed move leading to a state
  from which infinite runs are possible. These intervals are
  either of the form $\langle 0 , \gamma \cdot (b-y) \rangle$,
  or of the form $\langle \gamma \cdot (a-y) , \gamma \cdot (b-y) \rangle$
  where $\gamma$ is the cost rate of the corresponding state, and $a$,~$b$ are
  constants of the automaton.

  Now, we omit the details, but they are very similar to those for the original
  graph~$G$. In~this extended graph, the set of possible costs of paths in
  $\TTS_{\Aut}$ from~$(q,x)$ to~$(q',y)$ corresponds to the set of costs of
  paths in the new graph from~$(q,x,r)$ to~$\overline{(q',y,r')}$ and is a
  countable union
  \[
  \bigcup_{m \in \Nat} \langle \alpha_m - \beta_m \cdot x + \gamma_m \cdot y,
  \alpha'_m - \beta'_m \cdot x + \gamma'_m \cdot y \rangle
  \]
  where $\alpha_m$ and $\alpha'_m$ are integers (or $+\infty$), and $\beta_m$,
  $\beta'_m$, $\gamma_m$ and $\gamma'_m$ are costs of the automaton or $0$
  (result similar to Corollary~\ref{coro-8}). We~can even be more precise:
  $\beta_m$ is either~$0$ or the cost rate of~$q$, whereas $\beta'_m$ is the
  cost rate of~$q$. Similarly, $\gamma_m$ is either~$0$ or the cost rate
  of~$q'$, and $\gamma'_m$ is the cost rate of~$q'$. 

  \medskip
  A state $(q,x)$ will satisfy the formula $\Phi = \E\G[=c] \texttt{false}$
  whenever there is a run $\varrho$ in~$\Aut$ such that it can be decomposed
  into $\varrho = \varrho_1 \cdot \varrho_2 \cdot \varrho_3$ such that the
  cost of $\varrho_1$ is strictly less than~$c$, the cost of $\varrho_1 \cdot
  \varrho_2$ is strictly larger than~$c$
  and $\varrho_2$ corresponds to a single mixed move. 
  That is, whenever there exists a path from $(q,x,r)$
  to $\overline{(q',y,r')}$ of cost less than~$c$ s.t., when adding up the
  outgoing cost of a single mixed move, we get a cost larger than~$c$. As~in
  Lemma~\ref{lemma:gran}, we~can restrict the above union to a finite union,
  and we thus only need to solve finitely many linear systems of inequations.
  Then, we can analyze all possible cases for the bounds where the truth of
  $\Phi$ changes, and as previously, we see that the granularity needs only to
  be refined by $1/C$, hence the granularity which is required is~$1/C$
  (since we started from the classical region automaton, with
  non-refined constants).

  \medskip A state $(q,x)$ satisfies $\E\G[\geq c] \texttt{false}$ whenever
  there is an infinite run from $(q,x)$ for which the cost of all its prefixes
  is strictly less than $c$ (though the limit of these costs can be $c$
  itself). In such a run, there is a prefix of cost strictly less than $c$ and
  from that point on, the cost of each mixed move is very close to $0$ (and
  indeed as close as we want to $0$). We~thus proceed as follows: we~fix a
  location~$q$ and a region~$r$. For every $x$ and~$y$, we~compute the set of
  possible costs between $(q,x)$ and $(q,y)$ for $x,y \in r$. This is a
  countable union
  \[
  \bigcup_{m \in \Nat} \langle \alpha_m - \beta_m \cdot x + \gamma_m \cdot y,
  \alpha'_m \rangle
  \]
  after having simplified the previous union in which $\beta'_m$ and
  $\gamma'_m$ were both equal to the cost of location $q$. For each of the
  terms of the union, we distinguish between several cases:
  \begin{enumerate}[$\bullet$]
  \item if $\beta_m = \gamma_m = \alpha_m = 0$, then there is a cycle which
    can be iterated from $(q,r)$, and the global cost will be as small as we
    want. If the left-most bound of the interval is closed, then we can ensure
    a zero-cost, otherwise we cannot ensure a zero-cost.
  \item if $\beta_m = \gamma_m = 0$ but $\alpha_m > 0$, then there is no
    corresponding cycle that can be iterated without the cost to diverge.
  \item if $\beta_m = 0$ but $\gamma_m > 0$ is the cost of $q$, then the only
    chance to be able to iterate a cycle without paying too much is to choose
    $y$ be the left-most point $a$ of the region $r$. Then, either $\alpha_m +
    \gamma_m \cdot a = 0$, in which case we can iterate a cycle, or $\alpha_m
    + \gamma_m \cdot a > 0$, in which case we cannot iterate a cycle.
  \item if $\beta_m = 0$ but $\gamma_m > 0$ is the cost of $q$, a similar
    reasoning can be done, but with the right-most bound $b$ of $r$.
  \item if $\beta_m = \gamma_m > 0$ is the cost of location $q$, then it is
    not difficult to check that $\alpha_m$ is then not smaller than $\beta_m
    \cdot (b-a)$ (this can be checked on the graph $G$). Hence, a
    corresponding cycle can only be iterated if $a = b$, and thus if $r$ is a
    punctual region.
  \end{enumerate}
  The analysis of all these cases show that we only need to look at terms of
  the union such that $\alpha_m - \beta_m \cdot b + \gamma_m \cdot a = 0$, and
  either $a=b$, or the $\alpha_m \cdot \beta_m \cdot \gamma_m = 0$. Moreover,
  for each such constraint, it is only necessary to look at one of the
  witnessing intervals. We see that this set of states is a set of regions (we
  do not need to refine the region: a whole region either satisfies the
  property, or does not satisfy the property).

  That way, we can compute the set of states $S_0$ from which there exists an
  infinite run with a cost as small as possible (though possibly not zero).

  It remains to describe the set of states from which there is a finite path
  of cost strictly less than $c$ and reaching a state of $S_0$. This can
  easily be done using the extended graph $G$ we have presented above.

  \bigskip 
  \noindent$\blacktriangleright$ \ 
  \textbf{We now explain how we reduce all the other cases to the
    previous~ones.} We~consider the case of formula $\A\phi\U[\sim c]\psi$,
  still assuming that the automaton has no discrete costs. We~prove this
  result by reducing to the previous case. We consider the region automaton
  of~$\Aut$ w.r.t. constants $(a_i)_{0 \leq i \leq n+1}$ mentioned earlier
  (correct for subformulas $\varphi$ and~$\psi$), we~assume it is still a
  timed automaton (truth of formulas in the original automaton and in this
  region automaton is then equivalent).

  We moreover assume that we
  have two copies of each state, labeled with two extra atomic proposition
  $\textsf{has\_paid}$ and $\textsf{can\_have\_not\_paid}$
  which characterize when the last move had a positive cost, and when it could
  have no cost (for instance an instantaneous transition or a transition from
  a location where the cost rate is null).
  We denote the new automaton by~$\Aut_{\text{ext}}$, and give now a list of
  equivalences, not difficult to check, and useful for proving the induction step
  for  
  formulas of the form $\A\phi\U[\sim c]\psi$.

\begin{enumerate}[$\bullet$]
\item \((q,x),\Aut \models \A\phi\U[\geq c] \psi \ \text{ iff } \ (q,x),\Aut
  \models \A\phi \U \psi \et \A\G[<c] (\A\phi \U \psi) \et \A\F[\geq c]
  \texttt{true}\);

\item \((q,x),\Aut \models \A\phi\U[> c] \psi \ \text{ iff } \ (q,x),\Aut
  \models \A\phi \U \psi \et \A\G[\leq c] (\A\phi \U \psi) \et \A\F[>c]
  \texttt{true}\);

\item \((q,x),\Aut \models \E\G[>c]\texttt{false} \ \text{ iff } \ (q,x),\Aut
  \models \E\G[\geq c]\texttt{false} \ou \E\F[\leq
  c]\E\G(\textsf{can\_have\_not\_paid})\);


\item \((q,x),\Aut \models \A\phi\U[\leq c] \psi \ \text{ iff } \ (q,x),\Aut
  \models \A\phi \U \psi \et \A\F[\leq c] \psi\);

\item \((q,x),\Aut \models \E\G[\leq c] \psi \ \text{ iff } \
  (q,x),\Aut_{\text{ext}} \models \E\G\psi \ou \E\psi \U[ > c]\texttt{true}
  \);

\item \((q,x),\Aut \models \A\phi\U[< c] \psi \ \text{ iff } \ (q,x),\Aut
  \models \A\phi \U \psi \et \A\F[<c] \psi\);

\item \((q,x),\Aut \models \E\G[< c] \psi \ \text{ iff } \
  (q,x),\Aut_{\text{ext}} \models \E\G\psi \ou \E \psi \U[\geq
  c]\texttt{true}\);

\item \((q,x),\Aut \models \A\phi\U[=c] \psi \ \text{ iff } \ (q,x),\Aut
  \models \A\phi\U[\geq c]\psi \et \A\F[=c]\psi\);

\item \((q,x),\Aut \models \E\G[=c]\psi \ \text{ iff } $ \\
  \hfill $(q,x),\Aut_{\text{ext}} \models (\E\G[=c] \texttt{false}) \ou
  (\E\F[=c](\textsf{has\_paid} \wedge \psi \wedge (\E\G\psi \ou \E\psi\U
  \textsf{has\_paid})))\);
    
\end{enumerate}

Those transformations (which do not increase~$\height\Phi$) are sufficient to
lift the result to all the modalities of~\WCTL (under the assumption that we
have no discrete costs).

\bigskip

\subsubsection{We now explain how we can prove the induction step of
  Proposition~\ref{mainprop} for a formula $\Phi = \E\phi\U[\sim c]\psi$ when
  the automaton has discrete costs on transitions.}
  %
  %
  We will simplify the problem and reduce it to the computation of states
  satisfying a formula in an automaton without discrete costs. Then, applying
  the result proved for the automata without discrete costs, we will get the
  induction step.
  We note~$T$ the set of transitions of~$\Aut$ that have a positive discrete
  cost. We unfold the automaton as follows: there is a copy of~\Aut for every
  integer smaller than or equal to~$c+1$. Copy of location~$q$ in the $i$-th
  copy is denoted~$q_{(i)}$. There is a transition from $q_{(i)}$
  to~$q'_{(j)}$ if: either~$i=j$ and there is a transition in~$\Aut$ from $q$
  to $q'$ not in $T$; or $j=i+k \leq c+1$ and there is a transition in $T$
  with discrete cost $k$ from $q$ to~$q'$; or $j=p+1$, $i+k > c+1$ and there
  is a transition in $T$ with discrete cost $k$ from $q$ to~$q'$. We note
  $\Aut_{\text{unf}}$ this unfolding. Then, 
  \begin{eqnarray*}
    (q,x),\Aut \models \E \phi
    \U[\sim c] \psi & \ \text{iff}\ & (q_{(0)},x),\Aut_{\text{unf}} \models \bigvee_{i \leq p+1}
    \E \phi \U[\sim c-i] (\psi \wedge \textsf{copy}_i)
  \end{eqnarray*}
  where $\textsf{copy}_i$ is an atomic proposition labeling all locations
  of~$\Aut_i$. The correctness of this construction is obvious. Now, applying
  the induction hypothesis on automata with no discrete cost on transitions,
  the granularity of regions required for model-checking each formula is
  $1/C^{\max(\height\phi,\height\psi)+1}$, the granularity for the original
  formula in~$\Aut$ is thus also $1/C^{\max(\height\phi,\height\psi)+1} =
  1/C^{\height\Phi}$, which proves the induction step also for automata with
  discrete costs on transitions.

  \medskip 
  Finally, this extension to automata with discrete costs can be adapted to
  modalities of the form~$\A\U$. We~omit the tedious details.
\forceqed
\end{proof}

\begin{remark}
  In the above proof, we have exhibited exponentially many constants $a_i$'s
  at which truth of the formula can change. We will show here that the
  exponential number of constants is unavoidable in general. Indeed, consider
  the one-clock \PPTA~$\Aut$ displayed on Figure~\ref{fig-bin}.
  \begin{figure}[!ht]
    \centering
    \begin{tikzpicture}
      {\everymath{\scriptstyle}
        \path[use as bounding box] (0,-1.2) -- (6,.9);
        \draw (0,0) node [draw,line width=.6pt,circle] (A) {$\dot{p} = 1$};
        \draw (2,.6) node [draw,line width=.6pt,circle] (B2) {$\dot{p} = 4$};
        \draw (2,-.6) node [draw,line width=.6pt,circle] (B1) {$\dot{p} = 2$};
        \draw (4,.6) node [draw,line width=.6pt,circle] (C2) {$\dot{p} = 2$};
        \draw (4,-.6) node [draw,line width=.6pt,circle] (C1) {$\dot{p} = 1$};
        \draw (6,0) node [draw,line width=.6pt,circle] (D) {$\dot{p} = 1$};
        \draw (8,0) node [draw,line width=.6pt,circle] (E) {$\dot{p} = 1$};
        \draw [-latex',line width=.6pt] (A) -- (B1) node [below,midway,sloped] {$x < 1$};
        \draw [-latex',line width=.6pt] (A) -- (B2) node [above,midway,sloped] {$x \geq 1$};
        \draw [-latex',line width=.6pt] (B1) -- (C1) node [above,midway] {$x=2$} node [below,midway] {$x:=0$};
        \draw [-latex',line width=.6pt] (B2) -- (C2) node [above,midway] {$x=2$} node [below,midway] {$x:=0$};
        \draw [-latex',line width=.6pt] (C1) -- (D) node [below,midway,sloped] {$x<2$};
        \draw [-latex',line width=.6pt] (C2) -- (D) node [above,midway,sloped] {$x<2$};
        \draw [-latex',line width=.6pt] (D) -- (E) node [above,midway] {$x=0$};
        \draw [-latex',line width=.6pt,rounded corners=8pt] (D) -- (6,-1.3) -- (0,-1.3) -- (A);}
      \draw (-.5,-.5) node {${a}$};
      \draw (6.5,-.5) node {${b}$};
      \draw (8.5,-.5) node {${c}$};
    \end{tikzpicture}
\caption{The one-clock \PPTA~$\mathcal A$}\label{fig-bin}
\end{figure}
Using a \WCTL formula, we will require that the cost is exactly $4$
between~$a$ and~$b$. That way, if clock $x$
equals~$x_0.x_1x_2x_3\ldots x_n\ldots$ (this is the binary
representation of a real in the interval $(0,2)$) when leaving~$a$,
then it will be equal to~$x_1.x_2x_3\ldots x_n\ldots$ in~$b$.  We
consider the \WCTL formula $\phi(X) = \E\Big((a\ou b)\U[=0](\non
a\et\E(\non b \U[=4](b\et X)))\Big)$, where~$X$ is a formula we will
specify.  Then formula~$\phi(\E\F[=0] c)$ states that we can go
from~$a$ to~$b$ with cost~$4$, and that $x=0$ when arriving in~$b$
(since we can fire the transition leading to~$c$). From the remark
above, this can only be true if $x=0$ or~$x=1$ in~$a$.  Now, consider
formula~$\phi(\E\F[=0] c \, \ou\, \phi(\E\F[=0] c))$. If it holds in
state~$a$, then state~$c$ can be reached after exactly one or two
rounds in the automaton, \emph{i.e.}, if the value of~$x$ is
in~$\{0,1/2,1,3/2\}$.  Clearly enough, nesting~$\phi$ $n$~times
characterizes values of the clocks of the form~$p/2^{n-1}$ where $p$
is an integer strictly less than $2^n$.
\end{remark}

\input{complexity}


%% file: fig-graph.tex
\begin{figure}[!ht]
\centering
\begin{tikzpicture}
\everymath{\scriptstyle}
\draw (1,0.5) node[draw,rounded corners=1mm] (Ax) {$q,x,\{0\}$};
\draw (4,0.5) node[draw,rounded corners=1mm] (Bx) {$q,x,\{a_i\}$};
\draw (6.5,0.5) node[draw,rounded corners=1mm] (Cx) {$q,x,(a_i;a_{i+1})$};
\draw (9,0.5) node[draw,rounded corners=1mm] (Dx) {$q,x,\{a_{i+1}\}$};
\draw (1,-.1) node[draw,rounded corners=1mm] (A'x) {$q',x,\{0\}$};
\draw (4,-.1) node[draw,rounded corners=1mm] (B'x) {$q',x,\{a_i\}$};
\draw (6.5,-.1) node[draw,rounded corners=1mm] (C'x) {$q',x,(a_i;a_{i+1})$};
\draw (9,-.1) node[draw,rounded corners=1mm] (D'x) {$q',x,\{a_{i+1}\}$};
\draw (1,-.7) node {...};
\draw[dotted,rounded corners=2mm] (.3,.8) -- (1.7,.8) -- (1.7,-1) -- (.3,-1) -- cycle;
\draw (4,-.7) node {...};
\draw[dotted,rounded corners=2mm] (3.2,.8) -- (4.8,.8) -- (4.8,-1) -- (3.2,-1) -- cycle;
\draw (6.5,-.7) node {...};
\draw[dotted,rounded corners=2mm] (5.5,.8) -- (7.5,.8) -- (7.5,-1) -- (5.5,-1) -- cycle;
\draw (9,-.7) node {...};
\draw[dotted,rounded corners=2mm] (8.1,.8) -- (9.9,.8) -- (9.9,-1) -- (8.1,-1) -- cycle;
\draw (1,-2) node[draw,rounded corners=1mm] (A) {$q,\{0\}$};
\draw (4,-2) node[draw,rounded corners=1mm] (B) {$q,\{a_i\}$};
\draw (6.5,-2) node[draw,rounded corners=1mm] (C) {$q,(a_i;a_{i+1})$};
\draw (9,-2) node[draw,rounded corners=1mm] (D) {$q,\{a_{i+1}\}$};
\draw (1,-3) node[draw,rounded corners=1mm] (A') {$q',\{0\}$};
\draw (4,-3) node[draw,rounded corners=1mm] (B') {$q',\{a_i\}$};
\draw (6.5,-3) node[draw,rounded corners=1mm] (C') {$q',(a_i;a_{i+1})$};
\draw (9,-3) node[draw,rounded corners=1mm] (D') {$q',\{a_{i+1}\}$};
\draw (1,-3.7) node {...};
\draw[dotted,rounded corners=2mm] (.3,-1.7) -- (1.7,-1.7) -- (1.7,-4.3) -- (.3,-4.3) -- cycle;
\draw (4,-3.7) node {...};
\draw[dotted,rounded corners=2mm] (3.2,-1.7) -- (4.8,-1.7) -- (4.8,-4.3) -- (3.2,-4.3) -- cycle;
\draw (6.5,-3.7) node {...};
\draw[dotted,rounded corners=2mm] (5.5,-1.7) -- (7.5,-1.7) -- (7.5,-4.3) -- (5.5,-4.3) -- cycle;
\draw (9,-3.7) node {...};
\draw[dotted,rounded corners=2mm] (8.1,-1.7) -- (9.9,-1.7) -- (9.9,-4.3) -- (8.1,-4.3) -- cycle;
\draw[dashed,rounded corners=4mm] (0,-1.5) -- (11,-1.5) -- (11,-5) -- (0,-5) -- cycle;
%
\draw[rounded corners=1mm,-latex'] (1,-1) -- (1,-1.7);
\draw[rounded corners=1mm,-latex'] (4,-1) -- (4,-1.7);
\draw[rounded corners=1mm,-latex'] (6.5,-1) -- (6.5,-1.7);
\draw[rounded corners=1mm,-latex'] (6.8,-1) -- ++(-90:3mm) -| (8.7,-1.7);
\draw[rounded corners=1mm,-latex'] (9,-1) -- (9,-1.7);
%
\draw[dashed] (1.7,-3) -- +(0:5mm);
\draw[dashed,latex'-] (3.2,-3) -- +(180:5mm);
\draw[rounded corners=2mm,dashed] (1.1,-4.3) -- ++(-90:3mm) -- +(0:8mm);
\draw[rounded corners=2mm,dashed,latex'-] (3.7,-4.3) -- ++(-90:3mm) -- +(180:8mm);
%
\draw[-latex'] (4.8,-3) -- (5.5,-3);
\draw[dashed] (9.9,-3) -- +(0:5mm);
\draw[rounded corners=2mm,-latex'] (4.1,-4.3) -- +(-90:3mm) -| (8.9,-4.3);
\draw[rounded corners=2mm,dashed] (9.1,-4.3) -- ++(-90:3mm) -- +(0:8mm);
%
\draw[rounded corners=2mm,-latex'] (3.9,-4.3) -- +(-90:5mm) -| (.9,-4.3);
\draw[rounded corners=2mm,-latex'] (6.5,-4.3) -- +(-90:5mm) -| (.9,-4.3);
\end{tikzpicture}
\caption{(Schematic) representation of the graph $G$ (intervals labeling
  transitions have been omitted to improve readability)}\label{fig-graph}
\end{figure}

%% file: complexity.tex
\subsection{Algorithms and Complexity}
\label{sec:algo}

In this section, we provide two algorithms for model-checking \WCTL on
one-clock \PPTA.  The first algorithm runs in \EXPTIME, whereas the second one
runs in \PSPACE, thus matching the \PSPACE lower bound. However, it is
easier to first explain the first algorithm, and then reuse part of it
in the second algorithm. 
Finally, we will pursue the example of
Subsection~\ref{subsec:example}
for illustrating our \PSPACE algorithm.

\subsubsection{An \EXPTIME Algorithm}

The correctness of the algorithm we propose for model-checking one-clock \PPTA
against \WCTL properties relies on the properties we have proved in the
previous section: if $\Aut$ is an automaton with maximal constant $M$,
writing~$C$ for the l.c.m. of all costs labeling a location, and if $\Phi$ is
a \WCTL formula of constrained size $n$ (the maximal number of nested
constrained modalities), then the satisfaction of~$\Phi$ is uniform on the
regions~$(m/C^n;(m+1)/C^n)$ with $m<M\cdot C^n$, and also on $(M;+\infty)$.
The idea is thus to test the satisfaction of $\Phi$ for each state of the form
$(q,k/2 C^n)$ for $0 \leq k \leq (M\cdot 2 C^n) +1$ (\textit{i.e.} at the
bounds and in the middle of each region).

To check the truth of $\Phi = \E\phi\U[\cost \sim c]\psi$ in state $(q,x)$
with~$x=k/2C^n$, we will non-deterministically guess a witness. Using
graph~$G$ that we have defined in Section~\ref{section:granularite}, we begin
with proving a ``small witness property'': 
\begin{lemma}
  Let~$s$ be the smallest positive cost in~\Aut, and~$C$ be the lcm of all
  positive costs of~\Aut. Let~$q$ be a location of~\Aut, and~$x\in \R^+$.
  Let~$\Phi=\E\phi\U[\sim c]\psi$ be a \WCTL formula of size~$n$. Then
  $(q,x)\models \Phi$ iff there exists a run in~\Aut, from~$(q,x)$ and
  satisfying~$\phi\U[\sim c]\psi$, and whose projection in~$G$ visits at
  most~$N=\lfloor c\cdot C^n/s\rfloor+2$ times each state of~$G$.
\end{lemma}

\begin{proof}
  Let~$\tau$ be a run in~$\Aut$, starting from~$(q,x)$ (with~$x=k/2C^n$ for some~$k$) and
  satisfying~$\phi\U[\sim c]\psi$. To~that run corresponds a path~$\rho$
  in the region graph, starting in~$(q,x)$. Consider a cycle in that path~$\rho$: either
  it has a global cost interval~$[0,0]$, in which case it can be removed and
  still yields a witnessing run; or it has a global cost interval of the
  form~$\langle a,b\rangle$ with~$b>0$. In that case, letting~$s$ be the
  smallest positive cost of the automaton, we know that~$b\geq s/C^n$. Now, if
  some state of~$G$ is visited (strictly) more than~$N=\lfloor c\cdot
  C^n/s\rfloor+2$ times along~$\rho$, we build a path~$\rho'$ from~$\rho$ by
  removing extraneous cycles, in such a way that each state of~$G$ is visited
  at most~$N$ times along~$\rho$ (and that~$\rho$ starts and ends in the same
  states). Since we assumed that~$\rho$ does not contain cycles with cost
  interval $[0;0]$, we know that the upper bound of the accumulated cost
  along~$\rho'$ is above~$c$. Also, the lower bound of the accumulated costs
  along~$\rho'$ is less than that of~$\rho$. Since $\rho$ ``contains'' a run
  witnessing~$\phi\U[\sim c]\psi$, the cost interval of $\rho$ contains a
  value satisfying $\sim c$, thus so does the cost interval of $\rho'$. In
  other words, $\rho'$ still contains a path witnessing~$\phi\U[\sim c]\psi$.
  This path can easily be lifted to a run in $\Aut$ satisfying the formula
  $\phi\U[\sim c]\psi$.
\end{proof}

Since a transition in~$G$ may correspond to a linear sequence of transitions
in~$\Aut$, we know that if~$(q,x)\models \E\phi\U[\sim c]\psi$, then there
exists a witness having at most exponentially many transitions in~$\Aut$.

We now describe our algorithm: assuming we have computed, for each
state~$q$ of~$\Aut$, the intervals of values of~$x$ where~$\phi$
(resp.~$\psi$) holds, we non-deterministically guess the successive states of
a path in~$\Aut$, checking that~$\phi$ holds after each action transition and
that the path reaches a~$\psi$-state after an action transition and with cost
satisfying~$\sim c$. 
This verification can be achieved in~\PSPACE (and can be made deterministic as
$\PSPACE=\NPSPACE$). 
Since we apply this algorithm for each state~$(q,k/2C^n)$ with $0 \leq k \leq
(M \cdot 2 C^n) + 1$, our global algorithm runs in deterministic exponential
time.

It is immediate to design a similar algorithm for formulas~$\E\G[\geq c]\false$
and $\E\G[=c]\false$. 
The other existential modalities are handled by reducing to those case,
as explained in Section~\ref{section:granularite}.

\subsubsection{A \PSPACE Algorithm}

The \PSPACE algorithm will reuse some parts of the previous
algorithm,
but it will improve on space performance by computing and storing only the minimal
information required: instead of computing the truth value of each subformula
in each state~$(q,k/2C^n)$, it~will only compute the information it really needs. 
Our method is thus similar in spirit to the
space-efficient, on-the-fly algorithm for TCTL presented
in~\cite{HKV96}.

We will then need, while guessing a witness for $\E\phi\U[\cost \sim c]\psi$,
to check that all intermediary states reached after an action transition 
satisfy formula~$\phi$. As~$\phi$~might
be itself a \WCTL formula with several nested modalities, we will fork a new
computation of our algorithm on formula~$\phi$ from each intermediary state.
The maximal number of threads running simultaneaously is at most the depth of
the parsing tree of formula~$\Phi$. When a thread is preempted we only need to
store a polynomial amount of information in order to be able to resume it.
Indeed, it is sufficient to store for each preempted thread a triple
$(\alpha,K,I)$ where $\alpha$ is a node of the region graph, $K$~records
the number of steps of the path we are guessing (we~know that
when~$\E\phi\U[\sim c]\psi$ holds, an exponential witness exists), 
and $I$~is an interval corresponding
to the accumulated cost along the path being guessed.

The algorithm thus runs as follows: we start by labeling the root of the tree
by $\alpha=(q,x,r)$, $K=0$ and $I=[0;0]$. Then we guess a sequence of
transitions in the region graph, starting
from $(q,x,r)$; when a new state $(q',r')$ is added, we increment the
value of~$K$ and update the value of the interval, as described in the previous
section. If we just fired an action transition, then either we fork an execution for
checking that~$\phi$ holds, or we check that 
the constraint $\cost\sim c$ can be satisfied by the new interval and we
verify that the new state satisfies~$\psi$ (by again forking a new execution). 

The number of nested guesses can be bounded by the depth of the parsing tree
of~$\Phi$, because when a new thread starts, it starts from a node in the
parsing tree that is a child of the previous node. 
Thus, the memory  needed in this algorithm
is the parsing tree of formula $\Phi$ with each node labeled by a tuple which
can be stored in polynomial space. This globally leads to a \PSPACE algorithm.

\begin{example}
  We illustrate our \PSPACE algorithm on our initial example, with
  formula~$\Phi = \non\E(\textsf{OK}\ \U[t\leq 8] (\textsf{Problem} \et
  \non\E\F[c<30] \textsf{OK}))$. We write~$g=1/C^2$
  for the resulting granularity as defined in Prop.~\ref{mainprop}, and
  consider a starting state, e.g.~$(\textsf{OK},x=mg)$.

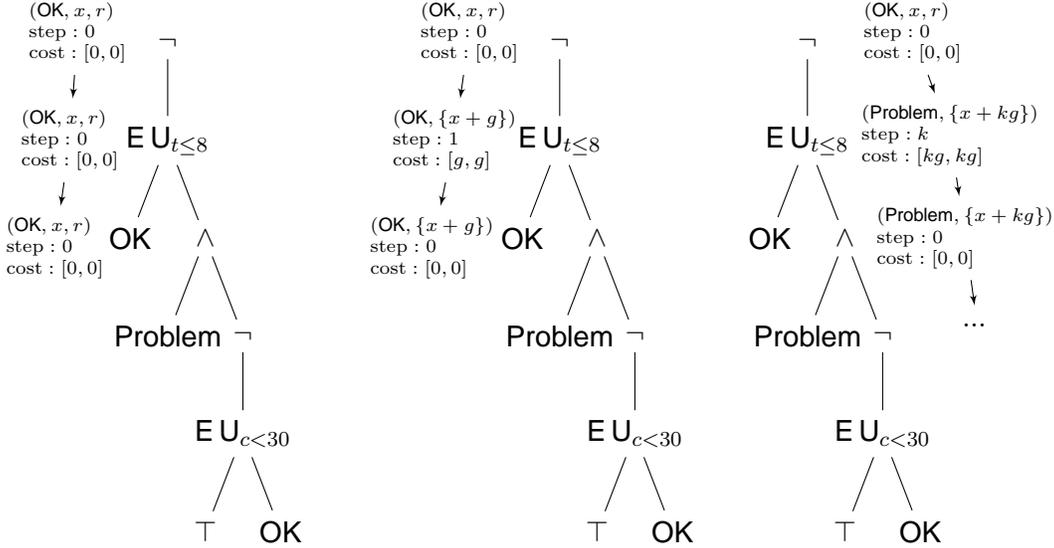
\begin{figure}[!ht]
\centering
\begin{tikzpicture}
\begin{scope}[yscale=1.3]
\path(1,3) node (non1) {$\non$} -- ++(-1.2,.1) node (Q) 
  {\tiny$\begin{array}{l}
    (\textsf{OK},x,r)\\
    \text{step}: 0\\
    \text{cost}: [0,0]\\
    \end{array}$};
\path(1,2) node (EUt) {$\E\U[t\leq 8]$} -- ++(-1.3,0) node (Q') 
  {\tiny$\begin{array}{l}
    (\textsf{OK},x,r)\\
    \text{step}: 0\\
    \text{cost}: [0,0]\\
    \end{array}$};
\path(.5,1) node (Ok1) {$\textsf{OK}$} -- ++(-1,-.1) node (Q'') 
  {\tiny$\begin{array}{l}
    (\textsf{OK},x,r)\\
    \text{step}: 0\\
    \text{cost}: [0,0]\\
    \end{array}$};
\path(1.5,1) node (et) {$\et$};
\path(1,0) node (Pb) {$\textsf{Problem}$};
\path(2,0) node (non2) {$\non$};
\path(2,-1) node (EUc) {$\E\U[c<30]$};
\path(1.5,-2) node (top) {$\top$};
\path(2.5,-2) node (Ok2) {$\textsf{OK}$};
\draw (non1) -- (EUt);
\draw (EUt) -- (Ok1);
\draw (EUt) -- (et);
\draw (et) -- (Pb);
\draw (et) -- (non2);
\draw (non2) -- (EUc);
\draw (EUc) -- (top);
\draw (EUc) -- (Ok2);
\draw[-latex'] (Q) -- (Q');
\draw[-latex'] (Q') -- (Q'');
\end{scope}
\end{tikzpicture}
\hfil
\begin{tikzpicture}
\begin{scope}[yscale=1.3]
\path(1,3) node (non1) {$\non$} -- ++(-1.2,.1) node (Q) 
  {\tiny$\begin{array}{l}
    (\textsf{OK},x,r)\\
    \text{step}: 0\\
    \text{cost}: [0,0]\\
    \end{array}$};
\path(1,2) node (EUt) {$\E\U[t\leq 8]$} -- ++(-1.4,0) node (Q') 
  {\tiny$\begin{array}{l}
    (\textsf{OK},\{x+g\})\\
    \text{step}: 1\\
    \text{cost}: [g,g]\\
    \end{array}$};
\path(.5,1) node (Ok1) {$\textsf{OK}$} -- ++(-1.2,-.1) node (Q'') 
  {\tiny$\begin{array}{l}
    (\textsf{OK},\{x+g\})\\
    \text{step}: 0\\
    \text{cost}: [0,0]\\
    \end{array}$};
\path(1.5,1) node (et) {$\et$};
\path(1,0) node (Pb) {$\textsf{Problem}$};
\path(2,0) node (non2) {$\non$};
\path(2,-1) node (EUc) {$\E\U[c<30]$};
\path(1.5,-2) node (top) {$\top$};
\path(2.5,-2) node (Ok2) {$\textsf{OK}$};
\draw (non1) -- (EUt);
\draw (EUt) -- (Ok1);
\draw (EUt) -- (et);
\draw (et) -- (Pb);
\draw (et) -- (non2);
\draw (non2) -- (EUc);
\draw (EUc) -- (top);
\draw (EUc) -- (Ok2);
\draw[-latex'] (Q) -- (Q');
\draw[-latex'] (Q') -- (Q'');
\end{scope}
\end{tikzpicture}
\hfil
\begin{tikzpicture}
\begin{scope}[yscale=1.3]
\path(1,3) node (non1) {$\non$} -- ++(1.4,.1) node (Q) 
  {\tiny$\begin{array}{l}
    (\textsf{OK},x,r)\\
    \text{step}: 0\\
    \text{cost}: [0,0]\\
    \end{array}$};
\path(1,2) node (EUt) {$\E\U[t\leq 8]$} -- ++(1.9,0.05) node (Q') 
  {\tiny$\begin{array}{l}
    (\textsf{Problem},\{x+kg\})\\
    \text{step}: k\\
    \text{cost}: [kg,kg]\\
    \end{array}$};
\path(.5,1) node (Ok1) {$\textsf{OK}$} ;
\path(1.5,1) node (et) {$\et$} -- ++(1.6,0) node (Q'') 
  {\tiny$\begin{array}{l}
    (\textsf{Problem},\{x+kg\})\\
    \text{step}: 0\\
    \text{cost}: [0,0]\\
    \end{array}$};
\path(1,0) node (Pb) {$\textsf{Problem}$};
\path(2,0) node (non2) {$\non$};
\path(2,-1) node (EUc) {$\E\U[c<30]$};
\path(1.5,-2) node (top) {$\top$};
\path(2.5,-2) node (Ok2) {$\textsf{OK}$};
\draw (non1) -- (EUt);
\draw (EUt) -- (Ok1);
\draw (EUt) -- (et);
\draw (et) -- (Pb);
\draw (et) -- (non2);
\draw (non2) -- (EUc);
\draw (EUc) -- (top);
\draw (EUc) -- (Ok2);
\draw[-latex'] (Q) -- (Q');
\draw[-latex'] (Q') -- (Q'');
\draw[-latex'] (Q'') -- +(-80:7mm) node[below=1mm]  {...};
\end{scope}
\end{tikzpicture}
\caption{Execution of our \PSPACE algorithm on the initial example.}\label{fig-run}
\end{figure}
Figure~\ref{fig-run} shows three steps of our algorithm. The first step
represents the first iteration, where subformula~$\textsf{OK}$ is satisfied at
the beginning of the run. At step~$2$, the execution goes to~$(\textsf{OK},
x+g)$: we check that the left-hand-side formula still holds in~$(\textsf{OK},
x+g)$ (as depicted), but also in intermediary states. The third figure
corresponds to $k$~steps later, when the algorithm decides to go to the
right-hand-part of~$\E\U[t\leq 8]$. In that case, of course, it is checked
that~$kg \leq 8$, and then goes on verifying the second until subformula.
\end{example}


%% file: WMTL.tex
\section{Model-checking linear-time logics}\label{sec:WMTL}

We now turn to the case of linear-time temporal logics. We begin with the
definition of our logic~\WMTL.

\subsection{The Logic \texorpdfstring{\WMTL}{WMTL}}

The logic \WMTL is a weighted extension of \LTL, but can also be
viewed as an extension of~\MTL~\cite{koymans90}, hence its name \WMTL,
holding for ``Weighted \MTL''. 

The syntax of \WMTL is defined inductively as follows:
\[
\WMTL \ni \varphi\ ::=\ a\ \mid\ \neg\varphi\ \mid\ \varphi \vee \varphi\
\mid\ \varphi \U[\cost \sim c] \varphi
\] 
where $a\in \AP$, $\cost$ is a cost function, $c$ ranges over~\Nat, $\mathord\sim
\in \{\mathord< , \mathord\le , \mathord= , \mathord\ge, \mathord> \}$.
If~there is a single cost function or if~the cost function~$\cost$ is clear
from the context, we~simply write $\varphi \U[\sim c] \psi$ for $\varphi
\U[\cost \sim c] \psi$.

We interpret \WMTL formulas over (finite) runs of labeled \PPTA, identifying
each cost of the formula with the corresponding cost in the automaton.

\begin{definition}
Let $\mathcal{A}$ be a labeled~\PPTA, and let $\varrho
= (q_0,v_0) \xrightarrow{\tau_1,e_1} (q_1,v_1)
\cdots \xrightarrow{e_p,\tau_p}
(q_p,v_p)$ be a finite run in~$\mathcal{A}$. 
The satisfaction relation for
\WMTL is then defined inductively as follows:
\begin{eqnarray*}
  \varrho \models a & \Leftrightarrow & a \in \ell(q_0) \\
  \varrho \models \neg \varphi & \Leftrightarrow  & \varrho \not\models \varphi \\
  \varrho \models \varphi_1 \ou \varphi_2 & \Leftrightarrow  & \varrho \models
    \varphi_1\ \text{or}\ \varrho \models \varphi_2 \\
  \varrho \models \varphi_1 \U[\cost \sim c] \varphi_2 & \Leftrightarrow  & \exists
  0<\pi\leq|\varrho| \ \text{s.t.}\ \varrho_{\geq \pi} \models \varphi_2,\
  \forall 0 < \pi'
  < \pi,\ \varrho_{\geq \pi'} \models \varphi_1, \\
  & & \phantom{\exists  0<\pi<|\varrho|\ \text{s.t.}\ } 
       \text{and}\ \cost(\varrho_{\leq \pi}) \sim c. 
\end{eqnarray*}
\end{definition}

\begin{example}
Back on our example of Figure~\ref{ex1}, we can express that there is no path
from \textsf{OK} back to itself in time less than~$10$ \emph{and} cost less
than~$20$. This is achieved by showing that no path satisfies the following formula:
\[
\textsf{OK}\, \U\, (\textsf{Problem}\, \et\, (\non\textsf{OK})\,\U[x\leq 10]\,
\textsf{OK}\, \et\, (\non \textsf{OK})\,\U[c\leq 20]\, \textsf{OK}).
\]
As we will see, model-checking \WMTL will in fact be undecidable when the
automaton involves more than one cost.
\end{example}

\begin{remark}
  Classically, there are two possible semantics for timed temporal
  logics~\cite{raskin99}: the continuous semantics, where the system is
  observed continuously, and the point-based semantics, where the system is
  observed only when the state of the system changes. We have chosen the
  latter, because the model checking problem for \MTL under the continuous
  semantics is already undecidable~\cite{AH90}, whereas the model-checking
  under the point-based semantics is decidable over finite runs~\cite{OW05}.
\end{remark}

\medskip We study existential model-checking of~\WMTL over priced timed
automata, stated as: given a one-clock \PPTA~$\mathcal A$ and a \WMTL
formula~$\varphi$, decide whether there exists a finite 
run~$\varrho$ in~$\mathcal{A}$ starting in an initial state and such that
$\varrho \models \varphi$.
Since \WMTL is closed under negation, our results obviously extend to the dual
problem of \emph{universal} model-checking.

We prove that the model-checking problem against \WMTL properties is decidable
for:
\begin{enumerate}
\def\labelenumi{(\theenumi)}
\item one-clock \PPTA with one stopwatch cost variable.
  \label{res:dec}
\end{enumerate}
Any extension to that model leads to undecidability. Indeed, we prove that the
model-checking problem against \WMTL properties is undecidable for:
\begin{enumerate}\addtocounter{enumi}{1}
\def\labelenumi{(\theenumi)}
\item one-clock \PPTA with one cost variable,
  \label{res:undec1}
\item two-clock \PPTA with one stopwatch cost variable,
  \label{res:undec2}
\item one-clock \PPTA with two stopwatch cost variables.
  \label{res:undec3}
\end{enumerate}

We present our results as follows. 
In Section~\ref{sec:dec}, we explain the positive
result~\eqref{res:dec} using an abstraction proposed in~\cite{OW05}
for proving the decidability of \MTL model checking over timed
automata.
Then, in Section~\ref{sec:undec}, we present all our undecidability
results, starting with the proof for result~\eqref{res:undec1}, and
then slightly modifying the construction for proving
results~\eqref{res:undec2}~and~\eqref{res:undec3}.

\subsection{Decidability of \texorpdfstring{\WMTL}{WMTL} for One-Clock \PPTA With One Stopwatch Cost}
\label{sec:dec}\label{dec}

\begin{theorem}
  Model checking one-clock \PPTA with one stopwatch cost against \WMTL
  properties is decidable, and non-primitive recursive.
\end{theorem}

\begin{proof}
  Time can be viewed as a special $\{1\}$-sloped cost. Hence, the
  non-primitive recursive lower bound follows from that of \MTL model
  checking over finite timed words, see~\cite{OW05,OW07}.

  The decidability then relies on the same encoding as~\cite{OW05}. We
  present the construction, but do not give all details, especially
  when there is nothing new compared with the above-mentioned paper.

  Let $\varphi$ be a~\WMTL formula, and $\mathcal{A}$ be a single-clock \PPTA
  with a stopwatch cost. Classically, from formula~$\varphi$, we construct an
  ``equivalent'' one-variable alternating timed automaton\footnote{We use the
    \emph{eager} semantics~\cite{BMOW07} for alternating automata,
    where 
    configuration of the automaton always have the same sets of successors.}
  $\mathcal{B}_\varphi$. Figure~\ref{fig-TAA} displays an example of such an
  automaton, corresponding to formula~$\G{}[a \thn (\F[\leq 3] b \ou \F[\geq
  2] c)]$ (see~\cite{OW05} for more details on alternating timed automata).
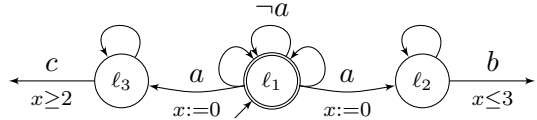
\begin{figure}[!ht]
\centering
\begin{tikzpicture}
\draw (0,0) node[circle,double,draw] (A) {$\scriptstyle \ell_1$};
\draw[latex'-] (A) -- +(-135:7mm);
\draw[-latex'] (A) .. controls +(60:10mm) and +(120:10mm) .. (A)
  node[midway,above] {$\non a$};
\draw (2,0) node[circle,draw] (B) {$\scriptstyle \ell_2$};
\draw[-latex'] (B) .. controls +(60:10mm) and +(120:10mm) .. (B)
  node[midway,above] {};
\draw[-latex'] (A) .. controls +(-10:10mm) and +(190:10mm) .. (B)
  node[midway,below] {$\scriptstyle x:=0$} node[midway,above] {$a$};
\draw[-latex'] (A) .. controls +(-10:10mm) and +(50:10mm) .. (A);
\draw (-2,0) node[circle,draw] (C) {$\scriptstyle \ell_3$};
\draw[-latex'] (C) .. controls +(60:10mm) and +(120:10mm) .. (C)
  node[midway,above] {};
\draw[-latex'] (A) .. controls +(190:10mm) and +(-10:10mm) .. (C)
  node[midway,below] {$\scriptstyle x:=0$} node[midway,above] {$a$};
\draw[-latex'] (A) .. controls +(190:10mm) and +(130:10mm) .. (A);
\draw[-latex'] (B) -- +(0:15mm) node[midway,above] {$b$} node[midway,below]
     {$\scriptstyle x\leq 3$};
\draw[-latex'] (C) -- +(180:15mm) node[midway,above] {$c$} node[midway,below]
     {$\scriptstyle x\geq 2$};
\end{tikzpicture}
\caption{A timed alternating automaton for formula $\G{}[a \thn (\F[\leq 3] b
    \ou \F[\geq 2] c)]$}\label{fig-TAA}
\end{figure}

However, note that in that case, the unique variable of the alternating
automaton is not a clock but a cost variable, whose rate will depend on the
location of~$\mathcal{A}$ which is being visited. However, as~for~\MTL,
we~have the property that $\mathcal{A} \models \varphi$ iff there is an
accepting joint run of~$\mathcal{A}$ and~$\mathcal{B}_\varphi$.

  In the following, we write $q$ for a generic location of
  $\mathcal{A}$ and $\ell$ for a generic location of
  $\mathcal{B}_\varphi$. Similarly, $Q$~denotes the set of locations
  of~$\mathcal{A}$ and $L$ the set of locations of~$\mathcal{B}_\varphi$.

  An \emph{$\mathcal{A}/\mathcal{B}_\varphi$-joint configuration} is a
  finite subset of $Q \times \R_{\geq 0} \cup L \times \R_{\geq 0}$
  with exactly one element of $Q \times \R_{\geq 0}$ (the current
  state in automaton $\mathcal{A}$). The joint behaviour of
  $\mathcal{A}$ and $\mathcal{B}_\varphi$ is made of time evolutions
  and discrete steps in a natural way.  Note that, from a given joint
  configuration $\gamma$, the time evolution is given by the current
  location $q_\gamma$ of $\mathcal{A}$: if the cost rate in $q_\gamma$
  is $1$, then all variables behave like clocks, \textit{i.e.}, grow
  with rate $1$, and if the cost rate in $q_\gamma$ is $0$, then all
  variables in $\mathcal{B}_\varphi$ are stopped, and only the clock
  of $\mathcal{A}$ grows with rate $1$.

  We encode configurations with words over the alphabet $\Gamma =
  2^{(Q \times \Reg \cup L \times \Reg)}$, where $\Reg = \{0,1,\ldots,
  M\} \cup \{\top\}$ ($M$ is an integer above the maximal constant
  appearing in both $\mathcal{A}$ and $\mathcal{B}_\varphi$). A state
  $(\ell,c)$ of $\mathcal{B}_\varphi$ will for instance be encoded by
  $(\ell,\mathit{int}(c))$ \footnote{$\mathit{int}$ represents the
    integral part.} if $c \leq M$, and it will be encoded by
  $(\ell,\top)$ if $c > M$.

  Now given a joint configuration $\gamma = \{(q,x)\} \cup
  \{(\ell_i,c_i) \mid i \in I\}$, partition~$\gamma$ into a sequence
  of subsets $\gamma_0,\gamma_1,\ldots,\gamma_p,\gamma_\top$, such
  that $\gamma_\top = \{(\alpha,\beta) \in \gamma \mid \beta > M\}$,
  and if $i,j \neq \top$, for all $(\alpha,\beta) \in \gamma_i$ and
  $(\alpha',\beta') \in \gamma_j$, $\mathit{frac}(\beta) \leq
  \mathit{frac}(\beta')$ \footnote{$\mathit{frac}$ represents the
    fractional part.} iff $i \leq j$ (so that $(\alpha,\beta)$ and
  $(\alpha',\beta')$ are in the same block $\gamma_i$ iff $\beta$ and
  $\beta'$ are both smaller than or equal to $M$ and have the same
  fractional part). We assume in addition that the fractional part of
  elements in $\gamma_0$ is $0$ (even if it means
  that~$\gamma_0=\varnothing$), and that all $\gamma_i$ for $1 \leq i
  \leq p$ are non-empty.

  If $\gamma$ is a joint configuration, we define its encoding
  $H(\gamma)$ as the word (over~$\Gamma$) 
  \[
  \reg{\gamma_0} \reg{\gamma_1} \ldots \reg{\gamma_p} \reg{\gamma_\top}
  \]
  where $\reg{\gamma_i} = \{(\alpha,\reg{\beta}) \mid (\alpha,\beta) \in
  \gamma_i\}$ with $\reg{\beta} = \mathit{int}(\beta)$ if $\beta \leq M$, and
  $\reg\beta=\top$ otherwise.
 
  \begin{example}
    Consider the configuration 
    \[
    \gamma = \{(q,1.6)\} \cup
    \{(\ell_1,5.2),(\ell_2,2.2),(\ell_2,2.6),(\ell_3,1.5),(\ell_3,4.5)\}.
    \] 
    Assuming that the maximal constant (on both~$\mathcal A$ and~$\mathcal
    B_{\varphi}$) is~$4$, the encoding is then
    \[
    H(\gamma) = \{(\ell_2,2)\} \cdot \{(\ell_3,1)\} \cdot
    \{(q,1),(\ell_2,2)\} \cdot \{(\ell_1,\top),(\ell_3,\top)\}
    \]
  \end{example}

  We define a discrete transition system over encodings of
  $\mathcal{A}/\mathcal{B}_\varphi$-joint configurations: there is a
  transition $W \Rightarrow W'$ if there exists $\gamma \in H^{-1}(W)$
  and $\gamma' \in H^{-1}(W')$ such that $\gamma \rightarrow \gamma'$
  (that can be either a time evolution or a discrete step).

  \begin{lemma}
    \label{lemma:bisimulation}
    The equivalence relation $\equiv$  defined as
    \(
    \gamma_1 \equiv \gamma_2 \eqdef H(\gamma_1) = H(\gamma_2)
    \)
    is a time-abstract bisimulation over joint configurations.
  \end{lemma}

  \begin{proof}
    We assume $\gamma_1 \rightarrow \gamma'_1$ and $\gamma_1
    \equiv \gamma_2$. We write $H(\gamma_1) = H(\gamma_2) = w_0 w_1
    \ldots w_p w_\top$ where $w_i \neq \emptyset$ if $1 \leq i \leq
    p$. We distinguish between the different possible cases for the
    transition $\gamma_1 \rightarrow \gamma'_1$.
    \begin{enumerate}[$\bullet$]
    \item assume $\gamma_1 \rightarrow \gamma'_1$ is a time evolution,
      and the cost rate in the corresponding location of $\mathcal{A}$
      is $0$. If $\gamma_1 = \{(q_1,x_1)\} \cup \{(\ell_{i,1},c_{i,1})
      \mid i \in I_1\}$, then $\gamma'_1 = \{(q_1,x_1+t_1)\} \cup
      \{(\ell_{i,1},c_{i,1}) \mid i \in I_1\}$ for some $t_1 \in
      \R_{\geq 0}$. We assume in addition that $\gamma_2 =
      \{(q_2,x_2)\} \cup \{(\ell_{i,2},c_{i,2}) \mid i \in I_2\}$.

      We set $\gamma^i_1$ the part of configuration $\gamma_1$ which
      corresponds to letter $w_i$, and we write $\alpha_1^i$ for the
      fractional part of the clock values corresponding to
      $\gamma^i_1$. We have $0 = \alpha_1^0 < \alpha_1^1 < \ldots <
      \alpha_1^p < 1$. We define similarly $(\alpha_2^i)_{0 \leq i
        \leq p}$ for configuration $\gamma_2$. We then distinguish
      between several cases:
      \begin{enumerate}[$-$]
      \item either $x_1+t_1 > M$, in which case it is sufficient to
        choose $t_2 \in \R_{\geq 0}$ such that $x_2+t_2 > M$.
      \item or $x_1+t_1 \leq M$ and $\mathit{frac}(x_1+t_1) =
        \alpha_1^i$ for some $0 \leq i \leq p$. In that case, choose
        $t_2 = x_1+t_1-\alpha_1^i+\alpha_2^i -x_2$. As $\gamma_1
        \equiv \gamma_2$, it is not difficult to check that $t_2 \in
        \R_{\geq 0}$. Moreover, $\mathit{frac}(x_2+t_2) = \alpha_2^i$
        and $\mathit{int}(x_2+t_2) = \mathit{int}(x_1+t_1)$.
      \item or $x_1+t_1 \leq M$ and $\alpha_1^i <
        \mathit{frac}(x_1+t_1) < \alpha_1^{i+1}$ for some $0 \leq i
        \leq p$ (setting $\alpha_1^{p+1} = 1$). As previously, in that
        case also, we can choose $t_2 \in \R_{\geq 0}$ such that
        $\alpha_2^i < \mathit{frac}(x_2+t_2) < \alpha_2^{i+1}$ and
        $\mathit{int}(x_2+t_2) = \mathit{int}(x_1+t_1)$.
      \end{enumerate}
      In all cases, defining $\gamma'_2 = \{(q_2,x_2+t_2)\} \cup
      \{(\ell_{i,2},c_{i,2}) \mid i \in I_2\}$, we get that $\gamma_2
      \rightarrow \gamma'_2$ and $\gamma'_1 \equiv \gamma'_2$, which
      proves the inductive case.
    \item there are two other cases (time evolution with rate of all
      variables being $1$, and discrete step), but they are similar to
      the case of \MTL, and we better refer to~\cite{OW07}.\qed
    \end{enumerate}
  \end{proof}

  Hence, from the previous lemma, we get:
  \begin{corollary}
    $W \Rightarrow^* W'$ iff there exist $\gamma \in H^{-1}(W)$ and
    $\gamma' \in H^{-1}(W')$ such that $\gamma \rightarrow^* \gamma'$.
  \end{corollary}

  The set $\Gamma = 2^{(Q \times \Reg \cup L \times \Reg)}$ is
  naturally ordered by inclusion $\subseteq$. We extend the classical
  subword relation for words over $\Gamma$ as follows: Given two words
  $a_0 a_1 \ldots a_n$ and $a'_0 a'_1 \ldots a'_{n'}$ in $\Gamma^*$,
  we say that $a_0 a_1 \ldots a_n \sqsubseteq a'_0 a'_1 \ldots
  a'_{n'}$ whenever there exists an increasing injection $\iota:
  \{0,1,\ldots,n\} \rightarrow \{0,1,\ldots,n'\}$ such that for every
  $i \in \{0,1,\ldots,n\}$, $a_i \subseteq
  a'_{\iota(i)}$. Following~\cite[Theorem~3.1]{AN00}, the
  preorder~$\sqsubseteq$ is a well-quasi-order.

  \begin{lemma}
    \label{lemma:pouet}
    Assume that $W_1 \sqsubseteq W_2$, and that $W_2 \Rightarrow^*
    W'_2$. Then, there exists $W'_1 \sqsubseteq W'_2$ such that $W_1
    \Rightarrow^* W'_1$.
  \end{lemma}

  The algorithm then proceeds as follows: we start from the encoding
  of the initial configuration, say~$W_0$, and then generate the tree
  unfolding of the implicit graph $(\Gamma^*,\Rightarrow)$, stopping a
  branch when the current node is labelled by~$W$ such that there
  already exists a node of the tree labelled by~$W'$ with $W'
  \sqsubseteq W$ (note that by Lemma~\ref{lemma:pouet}, if there is an
  accepting path from $W$, then so is there from $W'$, hence it is
  correct to prune the tree after node $W$).  Note that this tree is
  finitely branching. Hence, if the computation does not terminate,
  then it means that there is an infinite branch (by~K{\"o}nig
  lemma). This is not possible as $\sqsubseteq$ is a
  well-quasi-order. Hence, the computation eventually terminates, and
  we can decide whether there is a joint accepting computation
  in~$\mathcal{A}$ and~$\mathcal{B}_\varphi$, which implies that we
  can decide whether $\mathcal{A}$ satisfies~$\varphi$ or~not.
  \forceqed
\end{proof}

\begin{remark}
  In the case of~\MTL, the previous encoding can be used to prove the
  decidability of model checking for timed automata with any number of
  clocks. In our case, it cannot:
  Lemma~\ref{lemma:bisimulation} does not hold for two-clock \PPTA, even with
  a single stopwatch cost. Consider for instance two clocks $x$ and~$z$, and a
  cost variable~$\cost$. Assume we are in location~$q$ of the automaton with
  cost rate~$0$ and that there is an outgoing transition labelled by the
  constraint $x=1$. Assume moreover that the value of~$z$ is~$0$, whereas the
  value of~$x$ is $0.2$. We consider two cases: either the value of $\cost$ is
  $0.5$, or the value of~$\cost$ is~$0.9$. In both cases, the
  encoding\footnote{We extend the encoding we have presented above to several
    clocks, as originally done in~\cite{OW05}.} of the joint configuration is
  $\{(q,z,0)\} \cdot \{(q,x,0)\} \cdot \{(\cost,0)\}$. However, in the first
  case, the encoding when firing the transition will be $\{(q,x,1)\} \cdot
  \{(\cost,0)\} \cdot \{(q,z,0)\}$, whereas in the second case, it will be
  $\{(q,x,1)\} \cdot \{(q,z,0)\} \cdot \{(\cost,0)\}$. Hence the relation
  $\equiv$ is not a time-abstract bisimulation.
\end{remark}

\begin{remark}
  Let $\Aut$ be a \PPTA with a stopwatch cost. From the construction 
  using encodings by words we have presented above, we~see that truth of \WMTL
  formulas is invariant by classical regions (by~classical regions, we mean
  one-dimensional regions with granularity~$1$): indeed, in the above
  construction, it suffices to change the initial configuration with the
  encoding of the region we want to start from, and applying the previous
  results, we immediately get that the truth of the formula will then not
  depend on the precise initial value of the clock. 
  As~a consequence, the model checking of \WCTLstar\footnote{\WCTLstar is the
    extension of \CTLstar~\cite{CES86} with cost constraints. We omit its
    definition.} is decidable (and non-primitive recursive) for \PPTA with a
  single stopwatch cost: it suffices to label regions (in~the classical sense)
  with the \WMTL subformulas they satisfy.
  Let us mention right now that the undecidability results below directly extend to
  \WCTLstar, so that again, any extension of the model leads to undecidability. 
\end{remark}

\subsection{Undecidability Results}\label{sec:undec}\label{undec}

In this part, we prove that the above result is tight, in the sense that
adding an extra stopwatch cost ot removing the ``stopwatch'' condition yields
undecidability.

\subsubsection{One-Clock \PPTA With One Cost Variable}
 
\begin{theorem}\label{thm-1wta-WMTL}
  Model checking one-clock \PPTA with one (general) cost against \WMTL
  properties is undecidable.
\end{theorem}

We push some ideas used in~\cite{BBM06,BLM-fossacs07} further to prove this
new undecidability result.  We reduce the halting problem for a
two-counter machine $\mathcal{M}$ to that problem. The unique clock of
the automaton will store both values of the counters. If the first
(resp. second) counter has value $c_1$ (resp. $c_2$), then the value
of the clock will be $2^{-c_1} 3^{-c_2}$. Our machine~$\mathcal M$ has two
kinds of instructions. The first kind increments one of the counter,
say~$\mathtt{c}$, and jumps to the next instruction:
\begin{equation}
  \mathtt{p_i:\ c:=c+1;\ \texttt{goto}\ p_j.}\label{eq1}
\end{equation}
The second kind decrements one of the counter, say $\mathtt{c}$, and
goes to the next instruction, except if the value of the counter was
zero:
\begin{equation}
  \mathtt{p_i:\ \texttt{if}\ (c==0)\ \texttt{then goto } p_j \texttt{ 
      else } c:=c-1; \texttt{ goto }  p_k.} \label{eq2}
\end{equation}

Our reduction consists in building a one-clock \PPTA~$\mathcal A_{\mathcal M}$
and a \WMTL formula~$\phi$ such that the two-counter machine~$\mathcal M$
halts iff $\mathcal A_{\mathcal M}$ has a run satisfying~$\phi$. Each
instruction of~$\mathcal M$ is encoded as a module, all the modules are then
plugged together.

\paragraph{Module for instruction~\eqref{eq1}.}
Consider instruction~\eqref{eq1},
which increments the first counter. To simulate this instruction, we
need to be able to divide the value of the clock by~$2$.  The
corresponding module, named~$\Mod_i$, is depicted on
Figure~\ref{mod-incr1}.\footnote{As there is a unique cost variable, we
  write its rate within the location, and add a discrete
  incrementation (\textit{e.g.}~$+2$) on edges, when the edge has a
  positive cost.}
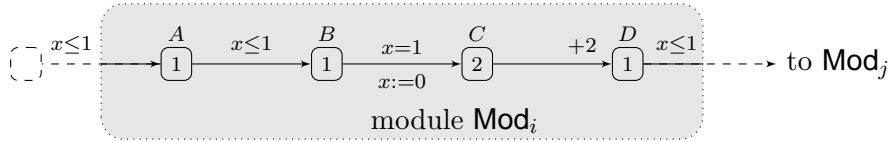
\begin{figure}[!ht]
  \centering
  \begin{tikzpicture}
    \everymath{\scriptstyle}
    \draw[rounded corners=3mm,dotted,fill=white!90!black] (-1,-1) -- (-1,.8) -- 
(7,.8) -- (7,-1) --cycle;
    \draw (0,0) node [draw,rounded corners=1mm] (A) {$1$} +(0,.4) node {$A$};
    \draw (2,0) node [draw,rounded corners=1mm] (B) {$1$} +(0,.4) node {$B$};
    \draw (4,0) node [draw,rounded corners=1mm] (C) {$2$} +(0,.4) node {$C$};
    \draw (6,0) node [draw,rounded corners=1mm] (D) {$1$} +(0,.4) node {$D$};
    \draw [-latex'] (A) -- (B) node [pos=.5,above] {$x \leq 1$};
    \draw [-latex'] (B) -- (C) node [midway,above] {$x=1$} node [midway,below] {$x:=0$};
    \draw [-latex'] (C) -- (D) node [near end,above] {${\scriptstyle +2}$};
    \draw (8.8,0) node (F) 
    {to~$\displaystyle \Mod_j$};
    \draw (-2,0) node (G) [draw,rounded corners=1mm,dashed] {$\phantom1$};
    \path (3.7,-.75) node {module~$\displaystyle \Mod_i$};
    \draw [-latex',dashed] (D) -- (F) 
    node[pos=.25,above] {$x\leq 1$}; 
    \draw (D) -- (7,0);
    \draw[latex'-,dashed] (A) -- (G) node[pos=.75,above] {$x\leq 1$};
    \draw (A) -- (-1,0);
  \end{tikzpicture}
  \caption{Module for incrementing~$c_1$}\label{mod-incr1}
\end{figure}

The following lemma is then easy to prove:
\begin{lemma}\label{lemme1}
  Assume that there is a run~$\rho$ entering
  module~$\Mod_i$ with~$x=x_0\leq 1$, exiting with~$x=x_1$, and such
  that no time elapses in~$A$ and~$D$ and the cost between~$A$ and~$D$
  equals~$3$.  Then $x_1=x_0/2$.
\end{lemma}

A similar result can be obtained for a module incrementing~$c_2$: it
simply suffices to replace the cost rate in~$C$ by~$3$ instead of~$2$.

\paragraph{Module for instruction~\eqref{eq2}.}
The simulation of this instruction is much more involved than the
previous instruction. Indeed, we first have to check whether the value
of~$x$ when entering the module is of the
form~$3^{-c_2}$ (\textit{i.e.}, whether~$c_1=0$). This is achieved,
roughly, by multiplying the value of~$x$ by~$3$ until it reaches
(or~exceeds)~$1$.  Depending on the result, this module will then
branch to module~$\Mod_j$ or decrement counter~$c_1$ and go to
module~$\Mod_k$. The difficult point is that clock~$x$ must be re-set
to its original value between the first and the second part.
We consider the module~$\Mod_i$ depicted on Figure~\ref{mod-decr1}.
\begin{figure}[!ht]
  \centering
  \begin{tikzpicture}
    \everymath{\scriptstyle}
    \draw[rounded corners=3mm,dotted,fill=white!90!black] (-.5,-2.7) -- (-.5,2) -- (9.3,2) -- (9.3,-2.7)
 --cycle;
    \draw (0,1) node [draw,rounded corners=1mm] (A0) {$1$} +(-.3,.3) node {$A_0$};
    \draw (1.3,1) node [draw,rounded corners=1mm] (B0) {$3$} +(0,.4) node {$B_0$};
    \draw (2.6,1) node [draw,rounded corners=1mm] (C0) {$1$} +(0,.4) node {$C_0$};
    \draw (0,0) node [draw,rounded corners=1mm] (A) {$1$} +(0,.4) node {$A$};
    \draw (1.3,0) node [draw,rounded corners=1mm] (B) {$3$} +(0,.4) node {$B$};
    \draw (2.6,0) node [draw,rounded corners=1mm] (C) {$1$} +(0.3,.3) node {$C$};
    %
    \draw (2.6,-1) node [draw,rounded corners=1mm] (C') {$1$} +(.4,0) node {$C'$};
    \draw (3.9,0) node [draw,rounded corners=1mm] (D) {$1$} +(0,.4) node {$D$};
    \draw (5.2,1) node [draw,rounded corners=1mm] (E1) {$3$} +(0,.4) node {$E_1$};
    \draw (5.2,-1) node [draw,rounded corners=1mm] (E2) {$3$} +(0,.4) node {$E_2$};
    \draw (6.2,1) node [draw,rounded corners=1mm] (F1) {$1$} +(0,.4) node {$F_1$};
    \draw (6.2,-1) node [draw,rounded corners=1mm] (F2) {$1$} +(0,.4) node {$F_2$};
    \draw (7.5,1) node [draw,rounded corners=1mm] (G1) {$3$} +(0,.4) node {$G_1$};
    \draw (7.5,-1) node [draw,rounded corners=1mm] (G2) {$3$} +(0,.4) node {$G_2$};
    \draw (8.5,1) node [draw,rounded corners=1mm] (H1) {$1$} +(0,.4) node {$H_1$};
    \draw (8.5,-1) node [draw,rounded corners=1mm] (H2) {$1$} +(0,.4) node {$H_2$};
    \draw (5,-2) node [draw,rounded corners=1mm] (A2) {$1$} +(0,-.4) node {$A_2$};
    \draw (6.3,-2) node [draw,rounded corners=1mm] (B2) {$2$} +(0,-.4) node {$B_2$};
    \draw (7.6,-2) node [draw,rounded corners=1mm] (C2) {$1$} +(0,-.4) node {$C_2$};
    \draw (8.9,-2) node [draw,rounded corners=1mm] (D2) {$1$} +(0,-.4) node {$D_2$};
    \draw[-latex'] (A0) -- (B0) node[midway,above=-2pt] {$x<1$};
    \draw[-latex'] (B0) -- (C0) node[midway,above=-2pt] {$x=1$} node[midway,below=-2pt] {$x:=0$};
    \draw[-latex'] (C0) -- (C);
    \draw[-latex'] (C) -- (D);
    \draw[-latex',rounded corners=1mm] (D) |- +(-1.5,-1.5) node[pos=.1,left] {$x<1$} -| (A);
    \draw[-latex'] (A) -- (B);
    \draw[-latex'] (C) .. controls +(-110:5mm) and +(110:5mm) .. (C');
    \draw[-latex'] (C') .. controls +(70:5mm) and +(-70:5mm) .. (C);
    \draw[-latex'] (B) -- (C) node[midway,above=-2pt] {$x=1$}
    node[midway,below=-2pt] {$x:=0$};
    \draw[-latex'] (D) -- (E1) node[midway,above=-2pt,sloped] {$x=1$}
    node[midway,below=-2pt,sloped] {$x:=0$}; 
    \draw[-latex'] (D) -- (E2) node[midway,above=-2pt,sloped] {$x>1$}
    node[midway,below=-2pt,sloped] {$x:=0$}; 
    \draw[-latex'] (E1) -- (F1);
    \draw[-latex'] (E2) -- (F2);
    \draw[-latex'] (F1) -- (G1) node[midway,below=-2pt] {$x:=0$};
    \draw[-latex'] (F2) -- (G2) node[midway,below=-2pt] {$x:=0$};
    \draw[-latex'] (G1) -- (H1);
    \draw[-latex'] (G2) -- (H2);
    \draw[-latex',rounded corners=1mm] (H2) |- +(-1,-.5) -| (A2);
    \draw[-latex'] (A2) -- (B2);
    \draw[-latex'] (B2) -- (C2) node[midway,above=-2pt] {$x=1$} node[midway,below=-2pt] {$x:=0$};
    \draw[-latex'] (C2) -- (D2) node[pos=.7,above=-2pt] {$\scriptscriptstyle +1$};
    \draw[dashed,-latex'] (D2) -- (9.7,-2) node[right] {to~$\displaystyle\Mod_k$};
    \draw (D2) -- (9.3,-2); 
    \draw[dashed,-latex'] (H1) -- (9.7,1) node[right] (lj) {to~$\displaystyle\Mod_j$};
    \draw (H1) -- (9.3,1); 
    \draw (-1.6,1) node (Z) [draw,rounded corners=1mm,dashed] {$\phantom1$};
    \draw[dashed] (Z) -- (A0) node[pos=.4,above=-2pt] {$x\leq 1$};
    \draw[-latex'] (-.5,1) -- (A0);
    \draw[-latex',dashed,rounded corners=1mm] (A0) |- +(.8,.7) -| (lj);
    \draw[rounded corners=1mm] (A0) |- +(9.3,.7) node[pos=.2,right=-2pt] {$x=1$};
    \path (2,-2.5) node {module~$\displaystyle \Mod_i$};
  \end{tikzpicture}
  \caption{Module testing/decrementing~$c_1$}\label{mod-decr1}
\end{figure}
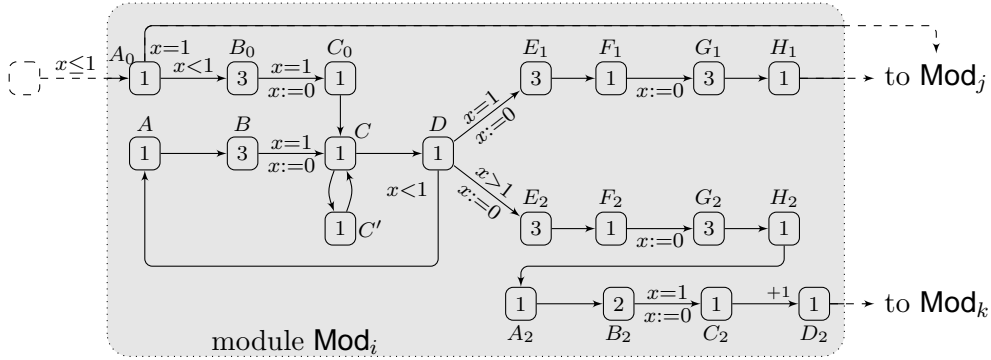

\begin{lemma}\label{lemme2}
  Assume there exists a run~$\rho$ entering module~$\Mod_i$
  with~$x=x_0\leq 1$, exiting to module~$\Mod_j$ with~$x=x_1$, and
  such that
  \begin{enumerate}[$\bullet$]
  \item no time elapses in~$A_0$, $C_0$, $D$, $A$, $C'$, $F_1$
    and~$H_1$;
  \item any visit to~$C_0$ or~$C'$ is eventually followed (strictly)
    by a visit to~$C'$ or~$F_1$;
  \item the cost exactly equals~$3$ along each part of~$\rho$
    between~$A$ or~$A_0$ and the next visit in~$D$, between~$C_0$
    or~$C'$ and the next visit in~$C'$ or~$F_1$, and between the last
    visit to~$D$ and~$H_1$.
  \end{enumerate}
  Then $x_1=x_0$ and there exists $n\in\N$ s.t.~$x_0=3^{-n}$.
\end{lemma}

\begin{proof}
  Let~$\rho$ be such a run. First, if~$x_0=1$ and $\rho$ goes
  directly to module~$\Mod_j$, then the result immediately follows.

  Otherwise, $\rho$ visits~$D$ at least once. We prove inductively
  that, at the~$k$-th visit in~$D$, the value of~$x$ equals~$3^k x_0$
  (remember that no time can elapse in~$D$). The first part of~$\rho$
  between~$A_0$ and~$D$ is as follows\footnote{By contradiction, it can be
  proved that~$C'$ cannot be visited along that part of~$\rho$, since the cost
  between~$C_0$ and~$C'$ must be exactly~$3$.} 
  (the labels on the arrows represent the
  cost of the corresponding transition):
  \[
  (A_0,x_0) \xrightarrow{0} (B_0,x_0) \xrightarrow{3(1-x_0)} (B_0,1)
  \xrightarrow{0} (C_0,0) \xrightarrow{0} (C,0) \xrightarrow{\alpha}
  (C,\alpha) \xrightarrow{0} (D,\alpha).
  \] 
  The total cost, $3(1-x_0)+\alpha$, must equal~$3$. Thus~$\alpha=3
  x_0$.  A~similar argument shows that one turn in the loop (from~$D$ back to
  itself) also
  multiplies clock~$x$ by~$3$, hence the result.  Since $\rho$
  eventually fires the transition from~$D$ to~$E_1$, it must be the
  case that~$x_0=3^{-n}$ for some~$n\in\N$.

  \medskip We now prove that~$x_1=x_0$. The proof follows a similar
  line: we prove that at the~$k$-th visit to~$C_0$ or~$C'$, the value
  of~$x$ is~$(3^k-3) x_0$. This clearly holds when~$k=1$ (\emph{i.e.}, when we
  visit~$C_0$). Assuming that~$\rho$ eventually visits~$C'$, we consider 
  the part of~$\rho$ between~$C_0$ and the first visit to~$C'$:
  \begin{multline*}
    (C_0,0) \xrightarrow{0} (C,0) \xrightarrow{3x_0} (C,3x_0)
    \xrightarrow{0} (D,3x_0) \xrightarrow{0} (A,3x_0) \xrightarrow{0}
    (B,3x_0) \\
    (B,3x_0) \xrightarrow{3(1-3x_0)} (B,1) \xrightarrow{0} (C,0)
    \xrightarrow{\beta} (C,\beta) \xrightarrow{0} (C',\beta).
  \end{multline*}
  The cost of this part is $3-6x_0+\beta$, and must equal~$3$. Thus
  $\beta=6x_0$ as required.  A similar computation (considering each part
  of~$\rho$ between two consecutive visits to~$C'$) proves the
  inductive case.

  Now, consider the part from the last visit of~$C'$ to~$H_1$:
  \begin{multline*}
    (C',(3^n-3)x_0) \xrightarrow{0} (C,(3^n-3)x_0) \xrightarrow{3x_0}
    (C,3^nx_0)
    \xrightarrow{0} (D,3^nx_0) \xrightarrow{0} (E_1,0) \\
    (E_1,0) \xrightarrow{3\gamma} (E_1,\gamma) \xrightarrow{0}
    (F_1,\gamma) \xrightarrow{0} (G_1,0) \xrightarrow{3\delta}
    (G_1,\delta) \xrightarrow{0} (H_1,\delta).
  \end{multline*}
  Remember that~$3^nx_0=1$, which explains why the computation goes to~$E_1$
  instead of~$E_2$). 
  The cost between~$C'$ and~$F_1$ is
  $3x_0+3\gamma$, and equals~$3$. Thus~$\gamma=1-x_0$. Similarly, the
  cost between~$D$ and~$H_1$ is $3\gamma+3\delta$ and must equal~$3$,
  which proves that~$\delta$, which is precisely~$x_1$, equals~$x_0$.
\end{proof}

We have a similar result for a run going to module~$\Mod_k$:
\begin{lemma}\label{lemme3}
  Assume there exists a run~$\rho$ entering module~$\Mod_i$
  with~$x=x_0\leq 1$, exiting to module~$\Mod_k$ with~$x=x_1$, and
  such that
  \begin{enumerate}[$\bullet$]
  \item no time elapses in~$A_0$, $C_0$, $D$, $A$, $C'$, $F_2$ $H_2$,
    $A_2$ and~$D_2$;
  \item any visit to~$C_0$ or~$C'$ is eventually followed (strictly)
    by a visit to~$C'$ or~$F_2$;
  \item the cost exactly equals~$3$ along each part of~$\rho$
    between~$A$ or~$A_0$ and the next visit in~$D$, between~$C_0$
    or~$C'$ and the next visit in~$C'$ or~$F_2$, between the last visit to~$D$
    and~$H_2$, and between~$H_2$ and~$D_2$.
  \end{enumerate}
  Then $x_1=2x_0$ and for every~$n\in\N$, $x_0\not=3^{-n}$.
\end{lemma}

\begin{proof}
  The arguments of the previous proof still apply: the value of~$x$ at
  the~$k$-th visit to~$D$ is~$3^kx_0$. If~$x_0$ had been of the
  form~$3^{-n}$, then $\rho$~would not have been able to fire the
  transition to~$E_2$. Also, the value of~$x$ when~$\rho$ visits~$H_2$
  is precisely~$x_0$. The part from~$H_2$ to~$D$ is then as follows:
  \[
  (H_2,x_0) \xrightarrow{0} (A_2,x_0) \xrightarrow{0} (B_2,x_0)
  \xrightarrow{2(1-x_0)} (B_2,1) \xrightarrow{0} (C_2,0)
  \xrightarrow{\kappa} (C_2,\kappa) \xrightarrow{1} (D_2,\kappa).
  \]
  The cost of this part is~$2(1-x_0)+\kappa+1$, 
  so that~$x_1=\kappa=2x_0$. 
\end{proof}

Again, these results can easily be adapted to the case of an
instruction testing and decrementing~$c_2$: it suffices to
\begin{enumerate}[$\bullet$]
\item set the costs of states~$B_0$, $B$, $E_1$, $E_2$, $G_1$
  and~$G_2$ to~$2$,
\item set the cost of~$B_2$ to~$3$, 
\item set the discrete cost of $C_2\to D_2$ to~$0$ 
\item set the discrete costs of~$C\to D$, $G_1\to H_1$ and~$G_2\to
  H_2$ to~$+1$.
\end{enumerate}  

\paragraph{Global reduction.}

We now explain the global reduction: the automaton~$\mathcal A_{\mathcal
  M}$ is obtained by plugging the modules above following the
instructions of~$\mathcal M$. There is one special module for
instruction~$\textsf{Halt}$, which is made of a single~$\textsf{Halt}$
state.  We also add a special initial state that lets $1$~t.u. elapse
(so that~$x=1$) before entering the first module.

The \WMTL formula is built as follows:
we first define an
intermediary subformula stating that no time can elapse in some given
state. It writes 
\(
\textsf{zero}(P) = \G(P \thn (P \U[=0] \non P))
\).
If the local cost in state~$P$ is not zero (which is the case in all
the states of~$\mathcal A_{\mathcal M}$), this formula forbids time
elapsing in~$P$.  We then let~$\phi_1$ be the formula requiring that
time cannot elapse in a state labelled with~$A$, $D$, $A_0$, $C_0$,
$C'$, $F_1$, $F_2$ $H_1$, $H_2$, $A_2$ and~$D_2$.
It remains to express the other conditions of
Lemmas~\ref{lemme1}, \ref{lemme2} and~\ref{lemme3}. 
We~write~$\phi_2$ for the corresponding formula..
For instance, the conditions of Lemmas~\ref{lemme2} and~\ref{lemme3} 
would be expressed as follows\footnote{The atomic
  proposition~$\Mod_{\text{decr}}$ is used to indicate that we are in a module
  decrementing one of the counters. It implicitly labels all the states of
  such modules.}:
\[
\G\left[
  A_0 \et \Mod_{\text{decr}} \thn 
  \left\{
  \begin{array}{c}
    \left(
      \begin{array}{l}
	(A\ou A_0) \thn (\non D \U[=3] D) \et\null\\
	(C_0\ou C') \thn (\non(C'\ou F_1) \U[=3] (C' \ou F_1)) \et\null \\
	(D\et \non D\U H_1) \thn (\non H_1 \U[=3] H_1)
      \end{array}
    \right) \U H_1
    \\\\
    \bigvee \\\\
    \left( 
      \begin{array}{l}
	(A\ou A_0) \thn (\non D \U[=3] D) \et\null\\
	(C_0\ou C') \thn (\non(C'\ou F_2) \U[=3] (C' \ou F_2)) \et\null \\
	(D\et \non D\U H_2) \thn (\non H_2 \U[=3] H_2) \et\null \\
	H2 \thn (\non D_2 \U[=3] D_2)
      \end{array}
    \right)\U H_2
  \end{array}
  \right\}
\right]
\]

The following proposition 
is now straightforward:
\begin{proposition}
  The machine~$\mathcal M$ halts iff there exists a run
  in~$\mathcal A_{\mathcal M}$ satisfying~$\phi_1\et\phi_2 \et
  \F\textsf{\upshape Halt}$.
\end{proposition}

\begin{remark}
  \begin{enumerate}[$\bullet$]
  \item For the sake of simplicity, our reduction uses discrete costs,
    so that our \WMTL formulas only involve constraints~``$=0$''
    and~``$=3$'' (and the same formula~$\varphi_2$ can be used for both
    counters). 
    But our undecidability result easily extends to automata without discrete costs.
  \item Our reduction uses a $\{1,2,3\}$-sloped cost variable, but it
    could be achieved with any $\{p,q,r\}$-sloped cost variable
    (with~$0<p<q<r$, and $p$, $q$ and~$r$ are pairwise coprime) by
    encoding the values of the counters by the clock value
    $(p/q)^{c_1}\cdot(p/r)^{c_2}$.
  \item Our \WMTL formula can easily be turned into a \WMITL formula
    (whose syntax is that of \MITL~\cite{AFH96}, \textit{i.e.}, with
    no punctual constraints). It suffices to replace formulas
    of the form $(\non p)\U[=n] p$ with $(\non p)\U[\leq n]p \et (\non
    p)\U[\geq n]p$.
  \end{enumerate}
\end{remark}

\subsubsection{Two-Clock \PPTA with One Stopwatch-Cost Variable}
While this case does not fit in our ``one-clock'' setting, it is an
interesting intermediate step between the previous and the next results.

\begin{theorem}
  Model checking two-clock \PPTA with one stopwatch cost against \WMTL
  properties is undecidable.
\end{theorem}

\begin{proof}
  The proof uses the same encoding, except that states with cost~$2$
  or~$3$ are replaced by sequences of states with costs~$0$ and~$1$
  having the same effect. We have two different kinds of states with
  cost~$2$ (or~$3$):
  \begin{enumerate}[$\bullet$]
  \item those in which we stay until~$x=1$:  
    \begin{center}
      \begin{tikzpicture}
        \everymath{\scriptstyle}
        \draw (0,0) node [dashed,draw,rounded corners=1mm] (A) {$\phantom1$} +(0,.4) node {$A$};
        \draw (2,0) node [draw,rounded corners=1mm] (B) {$2$} +(0,.4) node {$B$};
        \draw (4,0) node [dashed,draw,rounded corners=1mm] (C) {$\phantom1$} +(0,.4) node {$C$};
        \draw [dashed,-latex'] (A) -- (B) node [pos=.7,above] {$x \leq 1$};
        \draw [dashed,-latex'] (B) -- (C) node [pos=.3,above] {$x=1$} node [pos=.3,below] {$x:=0$};
        \draw (B) -- +(180:10mm);
        \draw (B) -- +(0:10mm);
      \end{tikzpicture}
    \end{center}
    These states are replaced by the following submodule:
    \begin{center}
      \begin{tikzpicture}
        \everymath{\scriptstyle}
        \draw (0,0) node [dashed,draw,rounded corners=1mm] (A) {$\phantom1$} +(0,.4) node {$A$};
        \draw (2,0) node [draw,rounded corners=1mm] (B1) {$1$} +(0,.4) node {$B$};
        \draw (4,0) node [draw,rounded corners=1mm] (B2) {$0$} +(0,.4) node {$B$};
        \draw (6,0) node [draw,rounded corners=1mm] (B3) {$1$} +(0,.4) node {$B$};
        \draw (8,0) node [dashed,draw,rounded corners=1mm] (C) {$\phantom1$} +(0,.4) node {$C$};
        \draw [dashed,-latex'] (A) -- (B1) node [pos=.7,above] {$x \leq 1$} 
        node [pos=.7,below] {$z:=0$};
        \draw [-latex'] (B1) -- (B2) node [pos=.5,above] {$x=1$} 
        node [pos=.5,below] {$x:=0$};
        \draw [-latex'] (B2) -- (B3) node [pos=.5,above] {$z=1$} 
        node [pos=.5,below] {$z:=0$};
        \draw [dashed,-latex'] (B3) -- (C) node [pos=.3,above] {$x=1$} node [pos=.3,below] {$x:=0$};
        \draw (B1) -- +(180:10mm);
        \draw (B3) -- +(0:10mm);
      \end{tikzpicture}
    \end{center}
   A simple computation shows that both sequences have the same
   effect on clock~$x$ and induce the same cost. Of course, the case
   of cost~$3$ is handled by adding one more pair of states with
   costs~$0$ and~$1$.

  \item those in which we enter with~$x=0$ (and exit with~$x\leq 1$):
    \begin{center}
      \begin{tikzpicture}
        \everymath{\scriptstyle}
        \draw (0,0) node [dashed,draw,rounded corners=1mm] (A) {$\phantom1$} +(0,.4) node {$A$};
        \draw (2,0) node [draw,rounded corners=1mm] (B) {$2$} +(0,.4) node {$B$};
        \draw (4,0) node [dashed,draw,rounded corners=1mm] (C) {$\phantom1$} +(0,.4) node {$C$};
        \draw [dashed,-latex'] (A) -- (B) node [pos=.7,below] {$x:=0$};
        \draw [dashed,-latex'] (B) -- (C) node [pos=.7,above] {$x\leq 1$};
        \draw (B) -- +(180:10mm);
        \draw (B) -- +(0:10mm);
      \end{tikzpicture}
    \end{center}    
    Those are replace with a slightly different sequence of states:
    \begin{center}
      \begin{tikzpicture}
        \everymath{\scriptstyle}
        \draw (0,0) node [dashed,draw,rounded corners=1mm] (A) {$\phantom1$} +(0,.4) node {$A$};
        \draw (2,0) node [draw,rounded corners=1mm] (B1) {$1$} +(0,.4) node {$B$};
        \draw (4,0) node [draw,rounded corners=1mm] (B2) {$0$} +(0,.4) node {$B$};
        \draw (6,0) node [draw,rounded corners=1mm] (B3) {$1$} +(0,.4) node {$B$};
        \draw (8,0) node [dashed,draw,rounded corners=1mm] (C) {$\phantom1$} +(0,.4) node {$C$};
        \draw [dashed,-latex'] (A) -- (B1) node [pos=.7,below] {$x:=0$};
        \draw [-latex'] (B1) -- (B2) node [pos=.5,above] {$x\leq 1$} node [pos=.5,below] {$z:=0$};
        \draw [-latex'] (B2) -- (B3) node [pos=.5,above] {$x=1$} 
        node [pos=.5,below] {$x:=0$};
        \draw [dashed,-latex'] (B3) -- (C) node [pos=.3,above] {$z=1$};
        \draw (B1) -- +(180:10mm);
        \draw (B3) -- +(0:10mm);
      \end{tikzpicture}
    \end{center}
  Again, one is easily convinced that both sequences are
  ``equivalent'', and that this transformation adapts to states with
  cost~$3$.\qed
  \end{enumerate}
\end{proof}

\subsubsection{One-Clock \PPTA with Two Stopwatch-Cost Variables}


In the above constructions, each clock can be replaced with an
observer variable, \textit{i.e.}, with a ``clock cost'' that is not
involved in the guards of the automaton anymore. We briefly explain
this transformation on an example, and leave the details to the keen
reader.

\begin{figure}[!ht]
  \centering
  \begin{minipage}{.45\linewidth}
    \begin{tikzpicture}
      \everymath{\scriptstyle}
      \draw (1,1) node[draw,rounded corners=1mm] (A) {$\phantom1$} +(0,.4) node {$A$};
      \draw (2.5,1) node[draw,rounded corners=1mm] (a) {$\phantom1$} +(0,.4) node {$\phantom B$};
      \draw (3.5,.5) node[draw,rounded corners=1mm] (b) {$\phantom1$} +(0,.4) node {$\phantom B$};
      \draw (3.5,1.5) node[draw,rounded corners=1mm] (c) {$\phantom1$} +(0,.4) node {$\phantom B$};
      \draw (5,.5) node[draw,rounded corners=1mm] (B) {$\phantom1$} +(0,.4) node {$B$};
      \draw (5,1.5) node[draw,rounded corners=1mm] (C) {$\phantom1$} +(0,.4) node {$C$};
      \draw[-latex'] (A) -- (a) node[midway,above=-2pt] {$x:=0$};
      \draw[-latex'] (a) -- (b);
      \draw[-latex'] (a) -- (c);
      \draw[-latex'] (b) -- (B) node[midway,above=-2pt] {$x=1$};
      \draw[-latex'] (c) -- (C) node[midway,above=-2pt] {$x<1$};
    \end{tikzpicture}
  \end{minipage}\hfil
  \begin{minipage}{.45\linewidth}
    \begin{tikzpicture}
      \everymath{\scriptstyle}
      \draw (1,1) node[draw,rounded corners=1mm] (A) {$1$} +(0,.4) node {$A$};
      \draw (1.9,1) node[draw,rounded corners=1mm,scale=.7] (x0) 
            {$\scriptscriptstyle 1$} +(0,.2) node[scale=.7] {$x_0$};
      \draw (2.5,1) node[draw,rounded corners=1mm] (a) {$1$} +(0,.4) node {$\phantom B$};
      \draw (3.5,.5) node[draw,rounded corners=1mm] (b) {$1$} +(0,.4) node {$\phantom B$};
      \draw (3.5,1.5) node[draw,rounded corners=1mm] (c) {$1$} +(0,.4) node {$\phantom B$};
      \draw (4.1,1.5) node[draw,rounded corners=1mm,scale=.7] (x<1) 
            {$\scriptscriptstyle 1$} +(0,.2) node[scale=.7] {$x_{<1}$};
      \draw (5,.5) node[draw,rounded corners=1mm] (B) {$1$} +(0,.4) node {$B$};
      \draw (4.1,.5) node[draw,rounded corners=1mm,scale=.7] (x=1) 
            {$\scriptscriptstyle 1$} +(0,.2) node[scale=.7] {$x_{=1}$};
      \draw (5,1.5) node[draw,rounded corners=1mm] (C) {$1$} +(0,.4) node {$C$};
      \draw[-latex'] (A) -- (x0); \draw[-latex'] (x0) -- (a); 
      \draw[-latex'] (a) -- (b);
      \draw[-latex'] (a) -- (c);
      \draw[-latex'] (b) -- (x=1); 
      \draw[-latex'] (x=1) -- (B);
      \draw[-latex'] (c) -- (x<1); 
      \draw[-latex'] (x<1) -- (C);      
    \end{tikzpicture}
  \end{minipage}
  \caption{Replacing a clock with an extra ``clock cost''}\label{fig-ex-obs}
\end{figure}
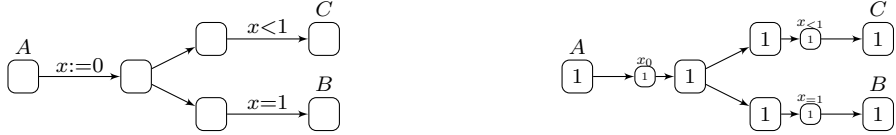
Figure~\ref{fig-ex-obs} displays the transformation to be applied to the
automaton. It then suffices to enforce that no time elapses in
states~$x_0$, $x_{<1}$ and~$x_{=1}$, and that the following formula
holds: 
\[
\bigwedge_{\mathord{\sim n} \in \{\mathord{<1}, \mathord{=1}\}} \G \Bigl[\bigl(x_0 \et \non x_0 \U x_{\sim n}\bigl)
\thn \bigl( \non x_0 \U[(c_x \sim n)] x_{\sim n}\bigl) \Bigr]
\]
This precisely encodes the role of clock~$x$ in the original automaton with a
clock cost, which is in particular a stopwatch cost. Note that this
transformation is not correct in general, but it is here because our reduction
never involves two consecutive transitions with the same guard. Thus, we get
immediately the following result:
\begin{theorem}
  Model checking one-clock \PPTA with two stopwatch-cost variables against
  \WMTL properties is undecidable.
\end{theorem}

%% file: conclu.tex
\section{Conclusion}

In this paper, we have studied various model-checking problems for one-clock
priced timed automata. We have proved that the model-checking of one-clock
priced timed automata against \WCTL properties is \PSPACE-complete. This is
rather surprising as model-checking \TCTL over one-clock timed automata has
the same complexity, though it allows much less features. For proving this
result, we have exhibited a sufficient granularity such that truth of formulas
over regions defined with this granularity is uniform.
Based on this result, we developed a space-efficient algorithm which
computes satisfaction of subformulas on-the-fly.
This result has to be contrasted with the undecidability result
of~\cite{BBM06} which establishes that model-checking priced timed automata
with three clocks and more against \WCTL properties is undecidable.

We have also depicted the precise decidability border for \WMTL
model-checking, a cost-constrained extension of \LTL. We have proved that the
restriction to single-clock single-stopwatch cost variable leads to
decidability, and that any single extension leads to undecidability.

There are several natural research directions: the decidability of
\WCTL model-checking for two-clocks priced timed automata is not
known, we just know that these models have an infinite
bisimulation~\cite{BBR04}; another interesting extension 
is multi-constrained modalities, 
e.g.~$\E\phi\U[\cost_1 \leq 5,\cost_2>3]\phi$?


%% file: biblio.tex
\newcommand{\etalchar}[1]{$^{#1}$}